\documentclass[a4paper, 11pt]{preprint}
\usepackage{amsthm, amsmath, amssymb, mathtools}
\usepackage[cal=boondoxo]{mathalfa}
\usepackage[full]{textcomp}
\usepackage[osf]{newtxtext}
\usepackage{mhequ}
\usepackage{mathrsfs}
\usepackage{microtype}
\usepackage{tikz}
\usepackage{enumitem}
\usepackage{comment}
\usepackage{orcidlink}
\usepackage{cprotect}
\usepackage{wasysym}
\usepackage{centernot}

\usepackage{hyperref}

\DeclareSymbolFont{timesoperators}{T1}{ptm}{m}{n}
\SetSymbolFont{timesoperators}{bold}{T1}{ptm}{b}{n}

\def\emptyset{\mathord{\centernot{\text{\rm$\Circle$}}}}

\usetikzlibrary{calc}
\usetikzlibrary{shapes}
\usetikzlibrary{decorations}
\usetikzlibrary{decorations.markings}
\usetikzlibrary{decorations.pathmorphing}

\newcommand{\eqdef}{\stackrel{\mbox{\rm\tiny def}}{=}}

\makeatletter
\renewcommand{\operator@font}{\mathgroup\symtimesoperators}
\makeatother

\usetikzlibrary{external}

\colorlet{darkred}{red!90!black}
\colorlet{darkblue}{blue!90!black}
\colorlet{lightblue}{blue!50}

\def\Wick#1{\mathord{{:}{#1}{:}}}

\newtheorem{theorem}{\bf  {Theorem}}
\newtheorem{proposition}[theorem]{\bf {Proposition}}
\newtheorem{lemma}[theorem]{\bf Lemma}

\newtheorem{corollary}[theorem]{\bf Corollary}
\newtheorem{remark}[theorem]{\bf Remark}

\numberwithin{theorem}{section}

\colorlet{symbols}{black}
\colorlet{testcolor}{green!60!black}

\tikzset{
	root/.style={circle,fill=testcolor,inner sep=0pt, minimum size=2mm},
	broot/.style={circle,fill=gray,inner sep=0pt, minimum size=2mm},
	dot/.style={circle,fill=black,inner sep=0pt, minimum size=1mm},
		reddot/.style={circle,fill=red,inner sep=0pt, minimum size=1mm},
			bluedot/.style={circle,fill=blue,inner sep=0pt, minimum size=1mm},
	eps/.style={circle,fill=white,draw=symbols,inner sep=0pt,minimum size=0.8mm},
	int/.style={circle,fill=black,draw=black,inner sep=0pt,minimum size=0.7mm},
	var/.style={circle,fill=black!10,draw=black,inner sep=0pt, minimum size=2mm},
	dotred/.style={circle,fill=black!50,inner sep=0pt, minimum size=2mm},
	generic/.style={semithick,shorten >=1pt,shorten <=1pt},
	dist/.style={ultra thick,draw=testcolor,shorten >=1pt,shorten <=1pt},
	testfcn/.style={ultra thick,testcolor,shorten >=1pt,shorten <=1pt,->},
	testfcnx/.style={ultra thick,testcolor,shorten >=1pt,shorten <=1pt,<-,
		postaction={decorate,decoration={markings,mark=at position 0.6 with {\drawx}}}},
	keps/.style={semithick,shorten >=1pt,shorten <=1pt,densely dashed,->},
	kprimex/.style={semithick,shorten >=1pt,shorten <=1pt,densely dashed,->,
		postaction={decorate,decoration={markings,mark=at position 0.4 with {\drawx}}}},
	kernel/.style={semithick,shorten >=1pt,shorten <=1pt,->},
	multx/.style={shorten >=1pt,shorten <=1pt,
		postaction={decorate,decoration={markings,mark=at position 0.5 with {\drawx}}}},
	kernelBig/.style={semithick,shorten >=1pt,shorten <=1pt,decorate, decoration={zigzag,amplitude=1.5pt,segment length = 3pt,pre length=2pt,post length=2pt}},
	kernelBig2/.style={semithick,shorten >=1pt,shorten <=1pt,decorate,->, decoration={snake,amplitude=1.5pt,segment length = 3pt,pre length=2pt,post length=5pt}},
	rho/.style={dotted,semithick,shorten >=1pt,shorten <=1pt},
	renorm/.style={shape=circle,fill=white,inner sep=1pt},
	labl/.style={shape=rectangle,fill=white,inner sep=1pt},
	H/.style={circle,fill=blue!10,draw=symbols,inner sep=0pt,minimum size=1.3mm},
	xi/.style={regular polygon sides=4,fill=red!30,draw=symbols,inner sep=0pt,minimum size=1.1mm},
	xix/.style={crosscircle,fill=symbols!10,draw=symbols,inner sep=0pt,minimum size=1.2mm},
	xib/.style={circle,fill=symbols!10,draw=symbols,inner sep=0pt,minimum size=1.6mm},
	xibx/.style={crosscircle,fill=symbols!10,draw=symbols,inner sep=0pt,minimum size=1.6mm},
	not/.style={circle,fill=symbols,draw=symbols,inner sep=0pt,minimum size=0.5mm},
Wick/.style={rectangle, draw=blue!80,rounded corners=3pt,fill=blue!5},
	>=stealth,
	not/.style={circle,fill=symbols,draw=symbols,inner sep=0pt,minimum size=0.5mm},
kernels2/.style={very thick,segment length=12pt},
	}
\makeatletter
\def\DeclareSymbol#1#2#3{%
	\expandafter\gdef\csname MH@symb@#1\endcsname{\tikzsetnextfilename{symbol#1}%
		\tikz[baseline=#2,scale=0.15,draw=symbols,line join=round]{#3}}%
	\expandafter\gdef\csname MH@symb@#1s\endcsname{\scalebox{0.75}{\tikzsetnextfilename{symbol#1}%
			\tikz[baseline=#2,scale=0.15,draw=symbols,line join=round]{#3}}}%
	\expandafter\gdef\csname MH@symb@#1ss\endcsname{\scalebox{0.65}{\tikzsetnextfilename{symbol#1}%
			\tikz[baseline=#2,scale=0.15,draw=symbols,line join=round]{#3}}}%
}
\def\<#1>{\ifthenelse{\boolean{mmode}}{\mathchoice{\csname MH@symb@#1\endcsname}{\csname MH@symb@#1\endcsname}{\csname MH@symb@#1s\endcsname}{\csname MH@symb@#1ss\endcsname}}{\csname MH@symb@#1\endcsname}}
\makeatother

\DeclareSymbol{out}{-3}{\draw[->] (0,0) node[dot]{}  -- (2,0);}
\DeclareSymbol{in}{-3}{\draw[<-,shorten <=1.5pt] (0,0) node[dot]{}  -- (2,0);}
\DeclareSymbol{N2}{-1}{ \draw (0,0) -- (0.8,1.2);\draw (0,0) -- (-0.8,1.2);\draw (0,0) node[dot]{}  -- (0,0);}
\DeclareSymbol{N3}{-3}{\draw (0,0) node[xi]{}  -- (0,0);}
\DeclareSymbol{N4}{-3}{\draw (0,0) node[H]{}  -- (0,0);}
\DeclareSymbol{one}{-3}{\draw (0,0) node[dot]{}  -- (0,0);}
\DeclareSymbol{N5}{-1}{\draw (-0.8,1.2) -- (0,0) node[H]{}  -- (0,0);}

\DeclareSymbol{ex}{3}{\draw (0,3) node[xi]{}  -- (-1,1.5) node[H]{}
	-- (0,0) node[dot]{} -- (1,1.5) node[H]{};}

\DeclareSymbol{xi}{-2}{\draw (0,0) node[xi,minimum size=1.6mm]{};}
\DeclareSymbol{H}{-2}{\draw (0,0) node[H,minimum size=1.8mm]{};}

\def\id{\mathrm{id}}

\def\CC{\mathcal{C}}

\def\CG{\mathcal{G}}
\def\CD{\mathcal{D}}
\def\CF{\mathcal{F}}
\def\CH{\mathcal{H}}

\def\CO{\mathcal{O}}

\def\CN{\mathcal{N}}

\def\CY{\mathcal{Y}}

\def\CX{\mathcal{X}}

\def\CU{\mathcal{U}}
\def\CE{\mathcal{E}}

\def\R{\mathbf{R}}
\def\C{\mathbf{C}}
\let\phi\varphi
\def\scal#1{\langle #1\rangle}
\def\${|\!|\!|}

\def\Vol{\operatorname{Vol}}
\let\eps\varepsilon

\def\det{\operatorname{det}}

\DeclareMathOperator{\tr}{tr}

\DeclareMathOperator{\Ob}{Ob}

\def\E{\mathbf{E}}
\def\P{\mathbf{P}}

\def\D{\mathbf{D}}

\def\out{o}
\def\inc{i}

\def\Hil{\mathbf{Hil}}

\def\f#1#2{\textstyle{\frac{#1}{#2}}}
\let\d\partial

\def\restr{\mathbin{\upharpoonright}}

\def\dash{\leavevmode\unskip\kern0.18em--\penalty\exhyphenpenalty\kern0.18em}
\def\slash{\leavevmode\unskip\kern0.15em/\penalty\exhyphenpenalty\kern0.15em}

\makeatletter

\DeclareRobustCommand{\TitleEquation}[2]{\texorpdfstring{\StrLeft{\f@series}{1}[\@firstchar]$\if%
		b\@firstchar\boldsymbol{#1}\else#1\fi$}{#2}}

\DeclareRobustCommand{\cev}[1]{%
  {\mathpalette\do@cev{#1}}%
}
\newcommand{\do@cev}[2]{%
  \vbox{\offinterlineskip
    \sbox\z@{$\m@th#1 x$}%
    \ialign{##\cr
      \hidewidth\reflectbox{$\m@th#1\vec{}\mkern7mu$}\hidewidth\cr
      \noalign{\kern-\ht\z@}
      $\m@th#1#2$\cr
    }%
  }%
}

\makeatother

\begin{document}
	\title{Probabilistic interpretation of quantum field theories\\[.5em]
	\small After Guillarmou, Kupiainen, Rhodes, Vargas, et al.}
	\author{Martin Hairer$^{1,2}$ \orcidlink{0000-0002-2141-6561}}
	\institute{École Polytechnique Fédérale de Lausanne, 1015 Lausanne, Switzerland 
	\and Imperial College London, London SW7 2AZ, United Kingdom\\
		\email{martin.hairer@epfl.ch}}
	
	\maketitle
	
	\begin{abstract}
In this note we provide a gentle introduction to the concepts and intuition
behind the recent breakthrough results on the mathematically rigorous construction 
of a non-trivial 
2D conformal field theory, namely the so-called Liouville theory. This gives us the
opportunity to review Segal's axioms for conformal field theories and to discuss in
some detail how the free field fits into them.\\[.5em]
\noindent These notes were written for the Séminaire Bourbaki given by the author in
Paris on 1 February 2025.

\vspace{1em}


	\end{abstract}
	
	\setcounter{tocdepth}{2}
	\tableofcontents
	
\section{Introduction}

In this note we provide a gentle introduction to a small selection of 
the concepts and intuition
behind the recent breakthrough results by Guillarmou, Kupiainen, Rhodes, Vargas, and
coauthors \cite{MR3465434,MR3832670,MR4029827,MR4060417,MR4749808,MR4816634,Bootstrap} 
(see in particular the review article \cite{Review}) on the mathematically rigorous construction and analysis of a completely integrable non-trivial 
2D conformal field theory, namely the so-called Liouville theory first introduced
by Polyakov in \cite{Poly1,Poly2,Poly3}.
We will spend some time on trying to understand
``$0$-dimensional QFT'', i.e.\ simply a single quantum mechanical 
particle, but from the perspective of rigorous path integration. 
This will naturally lead us to a form of Segal's axioms, which we then 
generalise to the two-dimensional conformal case.

We will then discuss in some detail how the Gaussian free field, reweighted by some 
local and coercive potential $V$ fits into this framework and satisfies Segal's axioms.
In the last section, we will finally introduce the case of Liouville theory.
Besides the functorial properties encoded in Segal's axions, this theory also
exhibits a form of conformal invariance which we discuss. In particular, we will
motivate the definition of the theory's central charge, as well as the Seiberg bounds
on the weights of its insertions. Due to a lack of both time and space, we'll leave
out most of the recent advances in this area, in particular the proof
of the DOZZ formula and the proof of the conformal bootstrap. The hope is rather that
after going through these notes the reader is equipped with some of the background
material required to read the original articles.

Most of the material of this note is based on \cite{Review,Bootstrap}, with significant
inspiration from \cite{SegalSimple}.

\section{The \TitleEquation{0}{0}-dimensional case}

Before we turn to quantum field theories, let us recall the path integral formulation
of a classical one-particle quantum system. This can be interpreted as a
``$0$-dimensional QFT'' and we will take this as a starting point to ``guess'' what a 
higher-dimensional QFT should look like. Throughout this article, we will 
only consider spin-less particles, which are then necessarily bosons.
Given a potential function $V \colon \R^d \to \R$ which we are going to assume 
smooth and bounded from below, the quantum mechanical Hamiltonian $H$ describing
the motion of a particle in the potential $V$ is the operator
\begin{equ}[e:defH]
H = -\f12\Delta + V\;,
\end{equ}
where $\Delta$ denotes the usual Laplacian on $\R^d$. We can realise $H$ as
a selfadjoint operator on $\CH = L^2(\R^d)$ by noting that the 
quadratic form 
\begin{equ}
B(\Phi,\Psi) = \int \bigl(\scal{\nabla \Phi(x),\nabla \Psi(x)} + V(x) \Phi(x)\Psi(x)\bigr)\,dx\;,
\end{equ}
defined  on $\CC_0^\infty \subset \CH$ is symmetric, positive and closable. Writing
$\CD(B)$ for the domain of its closure, it is a classical fact 
\cite{Friedrichs,Kato} that the  operator $H$ such that 
\begin{equ}[e:defDomain]
\CD(H) = \{\Phi \in \CD(B)\,|\, \exists \hat \Phi\in \CH \;\forall \Psi \in \CD(B)\,:\, \scal{\hat \Phi,\Psi} = B(\Phi,\Psi)\}
\end{equ}
and $H\Phi = \hat \Phi$ (with $\hat \Phi$ given as in \eqref{e:defDomain}, which
is unique by the density of $\CD(B)$ in $\CH$ and Riesz's representation theorem)
is indeed selfadjoint and agrees with \eqref{e:defH} on $\CC_0^\infty$.

Given such a Hamiltonian (which will in general be a selfadjoint operator that 
is bounded from below), a \textit{state} $\psi$ for the corresponding 
quantum-mechanical system
is a ray\footnote{namely a linear subspace of (complex) dimension one} in the complexification of the Hilbert space $\CH$. The evolution of such
a state is then given by the solution to the Schrödinger equation, namely
\begin{equ}
\d_t \psi = -iH\psi\;,
\end{equ}
where we identify $\psi$ with one of its representatives in $\CH$ (by linearity it doesn't
matter which one). This is of course nothing but the strongly continuous 
group of unitary operators
generated by the anti self-adjoint operator $-iH$. Conversely, given such a group, 
one can recover $H$ uniquely (but it need not be bounded from below).

\subsection{Path integral representation}
\label{sec:pathInt}

On the other hand, given $H$ as above, we can consider the \textit{semigroup}
$P_t^V$ generated by $-H$ and, conversely, any strongly continuous semigroup
of selfadjoint operators on $\CH$ is generated by a selfadjoint operator
that is bounded from below. Therefore, in order to identify $H$ and therefore 
to ``construct'' a quantum field theory, it is sufficient to 
construct the corresponding ``heat semigroup'' $P_t^V$ (since this is 
precisely what it is when $V\equiv 0$).
As a consequence of the Feynman--Kac formula \cite{Kac}, one has the 
following stochastic representation of $P_t^V$:
\begin{equ}[e:reprPt]
\bigl(P_t^V F\bigr)(x) = \E_x \biggl(F(\Phi_t)\exp \Bigl(-\int_0^t V(\Phi_s)\,ds\Bigr)\biggr)\;.
\end{equ}
Here, under the expectation $\E_x$, $\Phi$ is a standard Brownian motion
starting from the location $x$ at time $0$. The reason why we use this strange
notation (as opposed to $B$ or $W$) is that $\Phi$ will play the role of a field
later on.

A very fruitful idea that leads to the ``correct'' intuition is to formally rewrite \eqref{e:reprPt} 
as integral against the non-existing ``Lebesgue measure'' on the space of functions.
For this, recall that if $\mu$ is a Gaussian measure on $\R^N$ with covariance matrix $C$, then
it has a density with respect to Lebesgue measure given by
\begin{equ}[e:densityGaussian]
\mu(dx) = \frac{1}{\sqrt{\det(2\pi C)}} \exp \bigl(-\f12 \scal{x,C^{-1} x}\bigr)\,dx\;.
\end{equ}
In the case of Brownian motion, its covariance operator $C$ is the integral operator on
$L^2([0,T])$ (say) with kernel given by $C(s,t) = s \wedge t$.
If $\Phi \colon [0,T]\to \R$ is a smooth function 
with $\Phi(0) = 0$ and $\dot \Phi(T) = 0$, then an integration by parts on the first term shows that
\begin{equs}
\bigl(C\ddot \Phi\bigr)(t) &= \int_0^t s \ddot \Phi(s)\,ds + t\int_t^T \ddot \Phi(s)\,ds\\
&= t \dot \Phi(t) - \int_0^t \dot\Phi(s)\,ds - t \dot \Phi(t)
= - \Phi(t)\;.
\end{equs}
This shows that $C^{-1}$ is nothing but the operator $-\d_t^2$ with the abovementioned
boundary conditions. In view of \eqref{e:densityGaussian}, it is then very tempting to rewrite
\eqref{e:reprPt} as
\begin{equ}[e:pathInt]
\bigl(P_t^V F\bigr)(x) = Z^{-1}\int F(\Phi_t)\delta_x(\Phi_0)\exp \Bigl(-\int_0^t\bigl(\f12 |\dot \Phi_s|^2 + V(\Phi_s)\bigr)\,ds\Bigr)\,d\Phi\;,
\end{equ}
for some normalisation constant $Z$. This is of course completely nonsensical (for starters 
the constant $Z$
involves a factor $(2\pi)^\infty$ and the determinant of the unbounded operator $-\d_t^2$),
but it is nevertheless a powerful guide for our intuition of what a QFT should be.

In particular, the semigroup property of $P_t^V$ (which we recall is crucial in order to extract from 
it the Hamiltonian operator $H$) now appears naturally as a consequence of the fact that the 
expression appearing inside the exponential is the integral of a local expression of the field $\Phi$.
At the mathematically rigorous level, we recall that given $x,y \in \R^d$, a Brownian motion $\Phi$ 
starting at $\Phi_0 = x$ and conditioned to have $\Phi_t = y$ can be decomposed as
\begin{equ}[e:decompBM]
\Phi_s = \tilde \Phi_s + \frac{sy + (t-s)x}t\;,
\end{equ}
where $\tilde \Phi$ is a Brownian bridge. We deduce from this and the fact that $\Phi_t$ is a Gaussian
random variable with mean $x$ and variance $t$, that $P_t^V$ is an integral
operator with kernel given by 
\begin{equ}[e:kernel]
P_t^V(x,y) = P_t(x,y) \,\E \exp \biggl(-\int_0^t V\Bigl(\tilde \Phi_s+ \frac{sy + (t-s)x}t \Bigr)\,ds\biggr)\;,
\end{equ}
where $P_t$ denotes the usual heat kernel.
One nice feature of this representation is that, since we know that the operator $P_t$ with kernel
$P_t(x,y) \propto \exp(-\frac{|x-y|^2}{2t})$ is selfadjoint, we immediately see
from \eqref{e:kernel} that $P_t^V$ is also selfadjoint since the Brownian bridge measure is 
invariant under the change of variables $s \mapsto t-s$ which exchanges the roles of $x$ and $y$.

\subsection{Half-densities}
\label{sec:halfdens}

In order to study quantum field theories, one would like to generalise constructions of the type
\eqref{e:kernel} to infinite-dimensional situations. This however makes it somewhat unclear how expressions
like \eqref{e:kernel} in which $x$ and $y$ play symmetric roles can be extended to such a situation.
Indeed, one could ``naïvely'' think that the infinite-dimensional analogue of the ``heat kernel'' would 
be the Markov transition kernel of some infinite-dimensional Markov process (the analogue of the Brownian motion
in the previous discussion). Since however there isn't any analogue of Lebesgue measure in infinite dimensions,
a Markov kernel cannot be represented by a function of two variables there. Instead, it is naturally a function
in its first argument, but a measure in its second argument, thus breaking the nice symmetry between $x$ and $y$.

One solution to this problem is the use of so-called half-densities. These are based on the simple
observation that, given two positive Radon measures $\mu_1$, $\mu_2$ on a (Polish) space $\CX$, we
can canonically define a measure $\sqrt{\mu_1\mu_2}$ by
\begin{equ}
\sqrt{\mu_1\mu_2}(A) = \int_A \sqrt{\frac{d\mu_1}{d\nu}(x)\frac{d\mu_2}{d\nu}(x)}\,\nu(dx)\;,
\end{equ}
where $\nu$ is any positive Radon measure such that $\mu_i \ll \nu$ (for example $\nu = \mu_1 + \mu_2$)\footnote{Here and below we write $\mu \ll \nu$ to mean that $\mu$ is absolutely continuous with respect to $\nu$.}.
It is not difficult to prove (and already apparent from the notation) that this expression is
indeed independent of the choice of reference measure $\nu$. 

Given a measure class $[\nu]$ on $\CX$, we then have a Hilbert space $\CH_{[\nu]}$ which is formally 
nothing but $L^2(\CX,\nu)$, but we think of its elements as expressions of the type
$f \sqrt \nu$ with $f \in L^2(\CX,\nu)$, endowed with the scalar product
\begin{equ}
\scal{f \sqrt \nu , \tilde f \sqrt{\tilde \nu}} = \int_\CX f(x) \tilde f(x)\,\sqrt{\nu\tilde\nu}(dx) \;,
\end{equ}
as well as the natural equivalence relation postulating that $f_1 \sqrt{\nu_1} \sim f_2 \sqrt {\nu_2}$
if and only if there exists a measure $\mu$ with $\nu_i \ll \mu$ such that $f_1 \sqrt{\frac{d\nu_1}{d\mu}}
= f_2 \sqrt{\frac{d\nu_2}{d\mu}}$. In this way, $\CH_{[\nu]}$ is defined in a canonical way that only
depends on the measure class $[\nu]$ and not on its particular choice of representative $\nu$.

\begin{remark}\label{rem:tensor}
It is not difficult to see that these spaces have the same tensorial property as the usual  $L^2$
spaces, namely, given measure spaces $(\CX,\nu)$ and $(\CY,\mu)$ as above, we have
$\CH_{[\mu \otimes \nu]} \simeq \CH_{[\mu]} \otimes \CH_{[\nu]}$, with $\otimes$ denoting the tensor product 
of Hilbert spaces (which is again a Hilbert space). Here and below, we use the symbol
$\simeq$ to denote that two objects are not just isomorphic but \emph{canonically} isomorphic, so can be identified for all 
intents and purposes.
\end{remark}

We now remark that \eqref{e:kernel} can be written in a natural way as a half-density in the following way. Consider the $\sigma$-finite measure $\hat \P_t$ on $\CC([0,2t],\R^d)$
given by $\hat \P_t (d\Phi) = \int_{\R^d} \bigl(\tau^{(c)}_* \P^{(0)}_{2t}\bigr)(d\Phi)\,dc$, where
$\P^{(0)}_{2t}$ denotes the law of a Brownian bridge on $[0,2t]$ (which is therefore a probability measure on $\CC([0,2t],\R^d)$) and 
$(\tau^{(c)} \Phi)(s) = \Phi(s) + c$. We can then consider the measure $\hat \P^V_t$
given by 
\begin{equ}[e:exprPVhat]
\hat \P^V_t(d\Phi) = \exp \Bigl(-\int_0^{2t} V\bigl(\Phi_s\bigr)\,ds\Bigr)\,\hat \P_t(d\Phi)\;.
\end{equ}
If $V$ grows sufficiently fast at infinity (any strictly positive power
of its argument will do), then one can show that
the measure $\hat \P^V_t$ is finite.
Consider furthermore the map 
$\pi \colon \CC([0,2t],\R^d) \to \R^d \times \R^d$ given by 
\begin{equ}
\pi \Phi = (\Phi_0, \Phi_t)\;.
\end{equ}
We then claim the following.

\begin{proposition}
With $P_t^V$ as in \eqref{e:kernel} one has $P_t^V\sqrt{dx\,dy} = \sqrt{(4\pi t)^{-d/2}\pi_* \hat \P^V_t}$.
\end{proposition}

\begin{proof}
given $x,y \in \R^d$, write $\Phi_{x,y} \colon [0,2t] \to \R^d$ for the 
function that is affine on $[0,t]$ and $[t,2t]$ and such that 
$\Phi_{x,y} (0) = \Phi_{x,y} (2t) = x$ and $\Phi_{x,y}(t) = y$. 
We then consider the bijection
\begin{equs}
\Xi \colon \CC_0([0,t],\R^d)\times \R^d \times \R^d \times \CC_0([0,t],\R^d) &\to \CC([0,2t],\R^d)\\
(\Phi,x,y,\hat \Phi) &\mapsto \Phi_{x,y} + (\Phi \sqcup \hat \Phi)\;,
\end{equs}
where the concatenation $\Phi \sqcup \hat \Phi$ equals $\Phi$ on $[0,t]$
and $\hat \Phi(\cdot - t)$ on $[t,2t]$. Here, $\CC_0([0,t],\R^d)$ denotes
the space of continuous functions that vanish at each end of the interval.

With this notation at hand, it follows from a decomposition analogous to 
\eqref{e:decompBM} and the fact that the variance of $\Phi_t$ under $\P_{2t}^{(0)}$ is
$t/2$, that  
\begin{equ}
(\Xi^{-1})_* \hat \P_t = (\pi t)^{-d/2} e^{-\frac{|x-y|^2}{t}} \P^{(0)}_t \otimes dx \otimes dy \otimes \P^{(0)}_t\;.
\end{equ}
It follows immediately that
\begin{equs}
\frac{d\pi_* \hat \P_t^V}{dx\,dy} &= \frac{e^{-\frac{|x-y|^2}{t}}}{(\pi t)^{d/2} }\int \exp \Bigl(-\int_0^{2t} V\bigl(\Xi(\Phi,x,y,\bar \Phi)_s\bigr)\,ds\Bigr) \P_t^{(0)}(d\Phi)  \P_t^{(0)}(d\bar \Phi)\\
&= \frac{e^{-\frac{|x-y|^2}{t}}}{(\pi t)^{d/2} }\int \exp \Bigl(-\int_0^{t} V\bigl(\Phi_{x,y}(s) + \Phi_s\bigr)\,ds\Bigr) \P^{(0)}_t(d\Phi)\\
&\qquad\times  \int \exp \Bigl(-\int_0^{t} V\bigl(\Phi_{x,y}(s+t) + \bar \Phi_s\bigr)\,ds\Bigr) \P^{(0)}_t(d\bar \Phi) \\
&= \frac{\bigl(P_t(x,y)\bigr)^2}{(4\pi t)^{-d/2}} \biggl(\int \exp \Bigl(-\int_0^{t} V\bigl(\Phi_{x,y}(s) +  \Phi_s\bigr)\,ds\Bigr) \P^{(0)}_t(d \Phi)\biggr)^2\;,
\end{equs}
which is the desired identity.
\end{proof}

This is still not entirely satisfactory since it contains this strange factor $(4\pi t)^{-d/4}$ which appears to come out of nowhere. In fact, this can be understood by realising
two things. First, \eqref{e:exprPVhat} isn't very natural since it normalises $\hat\P_t$ to be a
probability measure while the exponential weight is unnormalised. In particular, we could
add a constant multiple  $\lambda$ of the identity to the inverse of the covariance 
operator of $\hat \P_t$ and simultaneously subtract $\frac{\lambda}2 \Phi^2$ to $V$. 
This would not change the formal expression \eqref{e:pathInt} which we took as a basis
for our intuition, but it would change the normalisation of the measure \eqref{e:exprPVhat}.
In view of \eqref{e:densityGaussian}, one would expect to be able to remedy this
if we were to multiply $\hat \P^V_t$ by $\sqrt{\det 2\pi C_t}$ for $C_t$ the covariance operator
of $\hat \P_t$. A discussion similar to that of Section~\ref{sec:pathInt} 
shows that $C_t^{-1} = -\d_t^2$, the second derivative 
operator on $L^2([0,2t],\R^d)$ with periodic boundary conditions.

Second, the formal expression ``$d\Phi$'' should represent ``Lebesgue measure'' in some
ambient Hilbert space $\CE_t$ of functions / distributions on $[0,2t]$. 
For this to concatenate in the ``right'' way, one should
have the decomposition $\CE_{t+s} = \CE_t \oplus \CE_s$, which is indeed the case
if one takes $\CE_t = L^2([0,2t])$. This however shows that, if we distribute $c$ according
to Lebesgue measure, then one should define $\tau^{(c)}$ appearing in the definition
of $\hat \P_t$ as $\tau^{(c)}\Phi = \Phi + \frac{c}{\sqrt{2t}}$, since it is 
$1/\sqrt{2t}$ which is a unit vector in $L^2([0,2t])$. 

Combining these two remarks and 
performing the change of variables $c \mapsto \sqrt{2t} c$ in the integral over $c$, 
this discussion shows that rather than considering $\hat \P^V_t$, it would be more
natural to consider
\begin{equ}[e:betterV]
\P^V_t =  \sqrt{\frac{2t}{\det (- \d_t^2/2\pi)}}\hat \P^V_t\;.
\end{equ}

\subsection{Determinants of differential operators}

Of course, taking the determinant of $-\d_t^2$ appears nonsensical at first sight.
However, this can be made sense of in the following way. For a symmetric 
strictly positive definite linear map
$A$  on  $\R^N$, we can define
\begin{equ}
\zeta_A(s) = \sum_{\lambda \in \sigma(A)} \lambda^{-s}\;.
\end{equ}
(We count eigenvalues with multiplicities here and below.)
For $\Re s$ large enough, this expression makes sense in much greater generality, 
for example (by Weyl's asymptotic) if $A$ is an elliptic selfadjoint differential 
operator on a compact manifold. 
Note then that
\begin{equ}
\zeta_A'(s) = -\sum_{\lambda \in \sigma(A)} \lambda^{-s} \log \lambda\;, 
\end{equ}
so that 
\begin{equ}[e:zetaA]
e^{-\zeta_A'(0)} = \exp \Bigl(\sum_{\lambda \in \sigma(A)} \log \lambda\Bigr)
= \prod \sigma(A) = \det A\;.
\end{equ}
This suggests that one \textit{defines} the $\zeta$-regularised determinant
of a self-adjoint operator $A$ by setting
\begin{equ}
\det_\zeta A = e^{-\zeta_A'(0)}\;,
\end{equ}
where we used the analytic continuation of $\zeta_A$ at the origin (provided that 
it exists).

\begin{remark}
Note that \eqref{e:zetaA} \textit{fails} when $A$ has some vanishing eigenvalues
since in that case $\det A$ vanishes while $\det_\zeta A$ equals
the determinant of the restriction of $A$ to its range.
\end{remark}

Take for example $A = -\d_t^2$ on the circle of radius $1$. In this case,
one has $\lambda_n = n^2$ with multiplicity $2$ and one gets
\begin{equ}
\zeta_A(s) = 2\zeta(2s)\;,
\end{equ}
where $\zeta$ is the usual Riemann zeta function. Since $\zeta'(0) = -\f12 \log(2\pi)$,
it follows that one has $\det_\zeta A = (2\pi)^2$. We now make use of the following
simple facts.

\begin{lemma}\label{lem:det}
One has $\det_\zeta(A\oplus B) = (\det_\zeta A)(\det_\zeta B)$ and,
for a scalar $\lambda$, one has $\det_\zeta(\lambda A) = \lambda^{\zeta_A(0)} \det A$.
\end{lemma}

\begin{proof}
One easily verifies that $\zeta_{A\oplus B}(s) = \zeta_A(s) + \zeta_B(s)$
whence the first claim follows. The second claim similarly follows from the fact that
$\zeta_{\lambda A}(s) = \lambda^{-s} \zeta_A(s)$.
\end{proof}

\begin{remark}
This shows that $\zeta_A(0)$ plays the role of an effective ``dimension'' 
for the range of $A$ which does not however need to be positive in general! Recall in 
particular that $\zeta(0) = -\f12$.
\end{remark}

We now remark that if $B$ is the operator $-(2\pi)^{-1}\d_t^2$ on $[0,2t]$ with 
periodic boundary conditions then, suitably interpreted, one has $B \sim \pi/(2t^2) A$,
where $\sim$ denotes unitary equivalence. It then follows from 
the second part of Lemma~\ref{lem:det} that
\begin{equ}
\det_\zeta B = \Bigl(\frac\pi {2t^2}\Bigr)^{\zeta_A(0)} \det_\zeta A
= (2\pi)^2 \Bigl(\frac\pi {2t^2}\Bigr)^{2\zeta(0)} = 8\pi t^2\;.
\end{equ}
In view of \eqref{e:betterV} and since $2t/\det_\zeta B = (4\pi t)^{-1}$, this suggests that a much more natural definition 
than \eqref{e:exprPVhat} is in fact given by
\begin{equ}[e:exprPV]
\P_t = (4\pi t)^{-d/2}\hat \P_t\;,\qquad 
\P^V_t(d\Phi) = \exp \Bigl(-\int_0^{2t} V\bigl(\Phi_s\bigr)\,ds\Bigr)\,\P_t(d\Phi)\;,
\end{equ}
which leads to the following statement.

\begin{corollary}\label{cor:SG}
With $P_t^V$ as in \eqref{e:kernel} one has $P_t^V\sqrt{dx\,dy} = \sqrt{\pi_*  \P^V_t}$.
\end{corollary}

\subsection{Cobordisms and Segal's axioms}
\label{sec:cobord1}

Let us now rewrite the semigroup property in the context of the 
measures $\P^V_t$ in a way that is suitable for extension to the case of QFTs.
One (admittedly extremely overkill in this situation) way of looking at it is the following
categorical viewpoint. Consider the category $\CC^{(1)}$ whose objects consists of finite
sets whose elements are labelled either ``out'' or ``in'', let's write
these as $A = A_\out \sqcup A_\inc$. Given such an $A$, we also write $\bar A$ 
for the ``opposite'' object, which is such that $\bar A_\out = A_\inc$ and $\bar A_\inc = A_\out$. 
We also define $A\sqcup B$ in the natural way.

Given such an $A$, a ``path above $A$'' is a compact oriented 
one-dimensional Riemannian manifold $\Sigma$ with boundary 
$\d\Sigma = \d\Sigma_\out \sqcup \d\Sigma_\inc$, as well as a bijection $\sigma \colon \d\Sigma \to A$ respecting the orientation. 
Connected components of $\Sigma$ are either isometric to an interval $[0,T]$
with $\sigma(0) \in A_\out$ and $\sigma(T) \in A_\inc$ or to a circle of some positive perimeter.
We include the degenerate case, so $\Sigma$ is allowed to contain intervals (but not circles) of length $0$ which we view as just one ``incoming'' and one ``outgoing'' point.
As a consequence, $A$ only admits paths over it if $|A_\out| = |A_\inc|$. Pictorially, 
we can draw elements of $A_\out$ as \<out> and elements of $A_\inc$ as \<in>.
Here is an example of a path over a set with one outgoing and one incoming point:
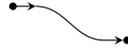
\begin{equ}
\begin{tikzpicture}[scale=0.15]
\draw[->] (0,0) node[dot]{}  -- (2,0);
\draw[<-,shorten <=1.5pt] (10,-3) node[dot]{}  -- ++(-2,0);
\draw(2,0) to[out=0,in=180] (8,-3);
\end{tikzpicture}
\end{equ}

A morphism $F \colon A \to B$ is then simply a path above $A\sqcup \bar B$, but
we impose the restriction that degenerate segments of length $0$ necessarily
connect elements of $A$ and $\bar B$ (as opposed to connecting two elements of $A$ say).
Two morphisms $F \colon A \to B$ and $G\colon B \to C$ can be
composed by gluing the two manifolds along the boundaries assigned to $B$.
(The points of $B$ that are incoming for $F$ are outgoing for $G$ and vice-versa,
so this respects orientation.) The condition imposed above guarantees that we
will never create degenerate loops in this process. 
The identity morphism $A\to A$ is the one consisting
solely of zero-length intervals connecting points of $A$ to themselves in $\bar A$.
In the following picture illustrating an example of composition of two morphisms,
we identify points of $A$ and $\bar A$ (say), but draw the arrows 
corresponding to $\bar A$ in red:
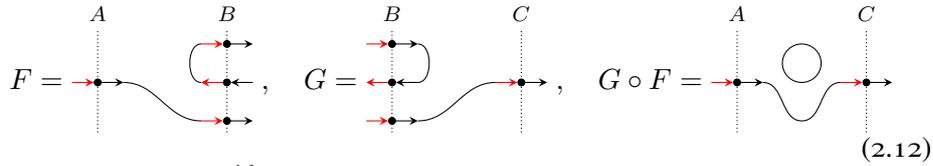
\begin{equ}[e:example]
F = \begin{tikzpicture}[scale=0.17,baseline=-0.6cm]
\draw[densely dotted] (0,1) node[above] {\scriptsize$B$} -- (0,-7);
\draw[densely dotted] (-10,1)node[above] {\scriptsize$A$} -- (-10,-7);

\draw[red,<-,shorten <=1.5pt] (0,0) node[dot]{}  -- ++(-2,0);
\draw[red,->] (0,-3) node[dot]{}  -- ++(-2,0);
\draw[red,<-,shorten <=1.5pt] (0,-6) node[dot]{}  -- ++(-2,0);

\draw[->] (0,0) node[dot]{}  -- ++(2,0);
\draw[<-,shorten <=1.5pt] (0,-3) node[dot]{}  -- ++(2,0);
\draw[->] (0,-6) node[dot]{}  -- ++(2,0);

\draw[->] (-10,-3) node[dot]{}  -- ++(2,0);
\draw[<-,red,shorten <=1.5pt] (-10,-3) node[dot]{}  -- ++(-2,0);

\draw(-8,-3) to[out=0,in=180] (-2,-6);
\draw(-2,0) to[out=180,in=180] (-2,-3);
\end{tikzpicture}
\;,\quad G= 
\begin{tikzpicture}[scale=0.17,baseline=-0.6cm]
\draw[densely dotted] (0,1) node[above] {\scriptsize$B$} -- (0,-7);
\draw[densely dotted] (10,1)node[above] {\scriptsize$C$} -- (10,-7);

\draw[red,<-,shorten <=1.5pt] (0,0) node[dot]{}  -- ++(-2,0);
\draw[red,->] (0,-3) node[dot]{}  -- ++(-2,0);
\draw[red,<-,shorten <=1.5pt] (0,-6) node[dot]{}  -- ++(-2,0);

\draw[->] (0,0) node[dot]{}  -- ++(2,0);
\draw[<-,shorten <=1.5pt] (0,-3) node[dot]{}  -- ++(2,0);
\draw[->] (0,-6) node[dot]{}  -- ++(2,0);

\draw[->] (10,-3) node[dot]{}  -- ++(2,0);
\draw[<-,red,shorten <=1.5pt] (10,-3) node[dot]{}  -- ++(-2,0);

\draw(2,-6) to[out=0,in=180] (8,-3);
\draw(2,0) to[out=0,in=0] (2,-3);
\end{tikzpicture}
\;,\quad G\circ F = 
\begin{tikzpicture}[scale=0.17,baseline=-0.6cm]
\draw[densely dotted] (0,1) node[above] {\scriptsize$A$} -- (0,-7);
\draw[densely dotted] (10,1)node[above] {\scriptsize$C$} -- (10,-7);

\draw[red,<-,shorten <=1.5pt] (0,-3) node[dot]{}  -- ++(-2,0);
\draw[->] (0,-3) node[dot]{}  -- ++(2,0);

\draw[->] (10,-3) node[dot]{}  -- ++(2,0);
\draw[<-,red,shorten <=1.5pt] (10,-3) node[dot]{}  -- ++(-2,0);

\draw(2,-3) to[out=0,in=180] (5,-6) to[out=0,in=180] (8,-3);
\draw(5,-1.5) circle (1.5);
\end{tikzpicture}
\end{equ}

If we set $A\otimes B \eqdef A \sqcup B$ and similarly define the product 
$F \otimes G$ of two morphisms as the disjoint union of the corresponding
collections of intervals, then this endows $\CC^{(1)}$ with the
structure of a monoidal category. Note that every morphism can be decomposed
into a finite product of elementary morphisms that can be of one of the 
following five types
\begin{equs}[e:elementary]
 I_t = 
\begin{tikzpicture}[scale=0.17,baseline=-0.1cm]
\draw[densely dotted] (0,1) -- (0,-1);
\draw[densely dotted] (6,1) -- (6,-1);

\draw[red,<-,shorten <=1.5pt] (0,0) node[dot]{}  -- ++(-2,0);
\draw[->] (0,0) node[dot]{}  -- ++(2,0);
\draw[->] (6,0) node[dot]{}  -- ++(2,0);
\draw[<-,red,shorten <=1.5pt] (6,0) node[dot]{}  -- ++(-2,0);

\draw(2,0) to[out=0,in=180] (4,0);
\end{tikzpicture}\;,\quad
I_t^* = 
\begin{tikzpicture}[scale=0.17,baseline=-0.1cm]
\draw[densely dotted] (0,1) -- (0,-1);
\draw[densely dotted] (6,1) -- (6,-1);

\draw[red,->] (0,0) node[dot]{}  -- ++(-2,0);
\draw[<-,shorten <=1.5pt] (0,0) node[dot]{}  -- ++(2,0);
\draw[<-,shorten <=1.5pt] (6,0) node[dot]{}  -- ++(2,0);
\draw[->,red] (6,0) node[dot]{}  -- ++(-2,0);

\draw(2,0) to[out=0,in=180] (4,0);
\end{tikzpicture}
\;,\\[1em]
C_t
= 
\begin{tikzpicture}[scale=0.17,baseline=-0.1cm]
\draw[densely dotted] (0,2) -- (0,-2);
\draw[densely dotted] (-4,2) -- (-4,-2);

\draw[red,<-,shorten <=1.5pt] (0,1) node[dot]{}  -- ++(-2,0);
\draw[red,->] (0,-1) node[dot]{}  -- ++(-2,0);

\draw[->] (0,1) node[dot]{}  -- ++(2,0);
\draw[<-,shorten <=1.5pt] (0,-1) node[dot]{}  -- ++(2,0);

\draw(-2,1) to[out=180,in=180] (-2,-1);
\end{tikzpicture}
\;,\quad 
C_t^*
= 
\begin{tikzpicture}[scale=0.17,baseline=-0.1cm]
\draw[densely dotted] (0,2) -- (0,-2);
\draw[densely dotted] (4,2) -- (4,-2);

\draw[red,<-,shorten <=1.5pt] (0,-1) node[dot]{}  -- ++(-2,0);
\draw[red,->] (0,1) node[dot]{}  -- ++(-2,0);

\draw[->] (0,-1) node[dot]{}  -- ++(2,0);
\draw[<-,shorten <=1.5pt] (0,1) node[dot]{}  -- ++(2,0);

\draw(2,1) to[out=0,in=0] (2,-1);
\end{tikzpicture}
\;,\quad 
O_t
= 
\begin{tikzpicture}[scale=0.17,baseline=-0.1cm]
\draw[densely dotted] (0,2) -- (0,-2);
\draw[densely dotted] (4,2) -- (4,-2);
\draw(2,0) circle (1.5);
\end{tikzpicture}\;,
\end{equs}
where the subscript $t \ge 0$ denotes the length of the single
interval (or circle) appearing in the corresponding morphism.
For example, the morphism $F$ appearing in \eqref{e:example}
is of the form $F =  I_t \otimes C_s$ for some $t \ge 0$ and $s > 0$.

Write now $\Hil$ for the category whose objets are (real) Hilbert spaces 
with morphisms consisting of bounded linear operators. This category
is also symmetric monoidal with $\otimes$ denoting the usual tensor products of
Hilbert spaces and linear operators. 

Given any Hilbert space $\CH$ and any semigroup $P_t$ of Hilbert--Schmidt
operators on $\CH$ (except of course at $t=0$ since the identity is not Hilbert--Schmidt),
we then obtain a monoidal functor
$\CF \colon \CC^{(1)} \to \Hil$ in the following way. 
Regarding objects, we set $\CF(A) = \CH^{\otimes A_\out} \otimes (\CH^*)^{\otimes A_\inc}$ (with the usual convention that $\CH^{\otimes \emptyset} = \R$). Of 
course, one has $\CH^* \simeq \CH$ since these are Hilbert spaces, but distinguishing the 
two is natural here since a Hilbert--Schmidt operator $P \colon \CH \to \CH$ can
 be viewed as an element $\bar P \in \CH^* \otimes \CH$. 
We then set $\CF(I_t) = P_t$, $\CF(C_t) = \bar P_t$, $\CF(O_t) = \tr P_t$, 
as well as $\CF(I_t^*) = \CF(I_t)^*$ and $\CF(C_t^*) = \CF(C_t)^*$.
We extend this to all morphisms by imposing that 
$\CF(F \otimes G) = \CF(F) \otimes \CF(G)$ and that it behaves ``correctly'' under
permutation of factors. It is then a straightforward exercise to show that 
the semigroup property is equivalent to the fact that $\CF$ is indeed a functor,
namely that $\CF(F\circ G) = \CF(F) \circ \CF(G)$.

In the particular situation of interest to us, one would choose $\CH$ to be the space
of half-densities
associated to the class of Lebesgue measure on $\R^d$ as in Section~\ref{sec:halfdens}.
Recall that if $V$ grows fast enough at infinity, then the measures $\P^V_t$
are finite, so that $\bar P_t = \sqrt{\pi_*\P^V_t} \in \CH\otimes \CH \simeq 
\CH^*\otimes \CH$. Furthermore, as a consequence of Corollary~\ref{cor:SG}, these 
operators do indeed form a semigroup, so that the above construction yields
a monoidal functor $\CF^V \colon \CC^{(1)} \to \Hil$.

\section{The free field as a conformal QFT}

We now aim to perform a similar construction for a quantum field theory, as opposed to
just a single particle. The approach described here is essentially the one
proposed by Segal in \cite{SegalOriginal} and is commonly referred to as ``Segal's approach''
to CFT's. This is also the version that was implemented in 
\cite{Bootstrap} in the context of Liouville CFT (which we haven't introduced yet!).

Instead of starting straight away with Liouville CFT, we start by interpreting the 
free field as a CFT. The problem with this is that, just as in the zero-dimensional
case described in the previous section, the free field has a zero mode that isn't 
pinned down, so that it is naturally described by a measure which is 
merely $\sigma$-finite as opposed to be finite. This would kick us out of our
framework, just like the previous construction fails when $V=0$ since the
measures $\P_t^V$ don't have finite mass in that case.
In order to circumvent this problem, our approach in this section will be to 
simply \textit{assume} that we 
are given a functional $V$ with suitable locality, coercivity, and invariance 
properties. The special case of Liouville theory will then be discussed in the 
last section.

\subsection{A general geometric setting}
\label{sec:cobord2}

In this section, we always work over a compact oriented Riemann surface $\Sigma$ with 
boundary $\d\Sigma$. Recall that $\Sigma$ is a one-dimensional complex manifold, 
i.e.\ such that each point in $\Sigma \setminus \d\Sigma$ admits a neighbourhood 
that is homeomorphic to the unit disc (viewed as a subset of $\C$) and such that 
transition maps are all holomorphic. 

This determines a collection of Riemannian metrics $g$ on $\Sigma$ which are those
that are proportional to the identity when viewed as metrics on $\C \simeq \R^2$ in any 
of the abovementioned charts. These are in particular such that,
for any two admissible metrics $g$, $\bar g$ there exists a smooth map $\phi \colon \Sigma \to \R$ such that $\bar g = e^{2\phi} g$. We furthermore assume that the boundary $\d\Sigma$
is oriented  and such that, for each connected component $\d\Sigma_j$ of $\d\Sigma$,
there exists a neighbourhood $\CN_j$ of $\d\Sigma_j$ and a holomorphic 
map $\psi_j \colon \CN_j \to \C$ such that $\psi(\d\Sigma_j)$ is the unit circle
(with the orientation of $\d\Sigma_j$ corresponding to the usual counterclockwise 
orientation) and there exists $\delta > 0$ such that either
$\psi_j(\CN_j) = \{z\,:\, |z| \in (e^{-\delta},1]\}$ or $\psi_j(\CN_j) = \{z\,:\, |z| \in [1,e^\delta)\}$. We write $\d\Sigma_\out$ for the union of the connected components such that
the former holds and $\d\Sigma_\inc$ for those such that it is the latter. 

We will always fix a parametrisation of $\d\Sigma$ and a metric $g$ on $\Sigma$ such
that the image of the parametrisation of $\d\Sigma_j$ under $\psi_j$ as above induces the
usual parametrisation $\theta \mapsto e^{i\theta}$ of the unit circle and such 
that the metric $g$ on each of the
$\CN_j$ agrees with the pullback of the flat metric $\frac{|dz|^2}{|z|^2}$ on $\C^*$ 
under $\psi_j$\footnote{The reason for using this metric and not simply the Euclidean 
one is that it is invariant under the reflection $z \mapsto \bar z^{-1}$. It is clearly flat
since it makes $\C^*$ isometric to $S^1 \times \R$ endowed with the Euclidean metric, by taking
logarithms.}. This in particular defines a notion of 
``outward normal derivative'' for smooth functions $F \colon \Sigma \to \R$, namely
$\d_\nu F(z) = {d\over dt} F(\psi_j^{-1}(\psi_j(z)e^t))\big\vert_{t=0}$
for $z \in \d\Sigma_j \cap \d\Sigma_\out$ and $\d_\nu F(z) = {d\over dt} F(\psi_j^{-1}(\psi_j(z)e^{-t}))\big\vert_{t=0}$
for $z \in \d\Sigma_j \cap \d\Sigma_\inc$.
(Note that, for some small $\delta > 0$, the function 
appearing on the right of this definition is defined on $(-\delta,0]$.)

Similarly to Section~\ref{sec:cobord1}, we can now define a category
$\CC^{(2)}$ of cobordisms in the following way. An object $A \in \Ob\CC^{(2)}$ is 
a finite (possibly empty) collection of circles, viewed as oriented one-dimensional Riemannian
manifolds of length $2\pi$. A morphism $\Sigma\colon A \to B$ is then 
a compact oriented Riemann surface with boundary as just described, together with
identifications $\Sigma_\inc \simeq A$ and $\Sigma_\out \simeq B$ that respect 
both the orientations and the metrics. We also call such a morphism a ``cobordism''.

The composition $\Sigma=\Sigma_2 \circ \Sigma_1$ of two cobordisms 
$\Sigma_1\colon A_1 \to A_2$ and $\Sigma_2\colon A_2 \to A_3$
is then obtained by simply ``gluing'' $\Sigma_1$ and $\Sigma_2$ along $\Sigma_{1,\out} \simeq A_2 \simeq \Sigma_{2,\inc}$.
As a set, one has $\Sigma = (\Sigma_1 \sqcup \Sigma_2) / A_2$, which is equipped with Riemannian and complex structures
inherited from $\Sigma_1$ and $\Sigma_2$ away from $A_2$. If $\d\Sigma_j$ is any of the connected components of 
$A_2$, then we have charts $\psi_{j,1} \colon \CN_{j,1} \to \C$ and $\psi_{j,2} \colon \CN_{j,2} \to \C$
as above, where the $\CN_{j,k}$ are some neighbourhoods of $\d\Sigma_j$
in $\Sigma_k$. We then simply concatenate them, yielding charts $\psi_{j} \colon \CN_{j} \to \C$,
where $\CN_j = \CN_{j,1} \cup \CN_{j,2}$ is a neighbourhood of $\d\Sigma_j$ in $\Sigma$. 

As before, we allow our cobordisms to contain degenerate components which 
simply consist of an identification of a connected component of their
domain with a connected component of the codomain. Again, these have to respect
the orientation and Riemannian structure. The identity morphism $\id \colon A\to A$
then consists solely of degenerate components and identifies $A$ with itself in the 
canonical way. However, we also have morphisms $\id_\theta \colon S^1 \to S^1$ 
consisting of the identification of $S^1$ with itself, rotated by an angle $\theta$.
In the case $A_2 = S^1$ and $\Sigma_i$ as above, $\Sigma_2 \circ \Sigma_1$
differs from $\Sigma_2 \circ \id_\theta \circ \Sigma_1$ by replacing
$\psi_{j,2}$ by $z\mapsto e^{i\theta} \psi_{j,2}(z)$ in the above construction.
In this way, $\CC^{(2)}$ is again a symmetric monoidal category with 
$\otimes$ being the natural disjoint union.

It also comes with two additional pieces of structure.
First, for any
$\Sigma\colon A \to B$, we have a morphism $\Sigma^* \colon B \to A$ obtained
from $\Sigma$ by reversing the orientation
of $\Sigma$, but not that of $\d\Sigma$, so that in the identification
$\d\Sigma^* = \d\Sigma$, one has $\d\Sigma^*_\out = \d\Sigma_\inc$
and vice-versa. 
More precisely, given a neighbourhood $\CN_j$ of a boundary component $\d\Sigma_j$ of $\d\Sigma_\out$ (say), 
as well as the corresponding chart $\psi_j \colon \CN_j \to \{z\,:\, |z| \in (e^{-\delta},1]\}$ of $\Sigma$,
we set $\psi_j^*(z) = (\overline{ \psi_j(z)})^{-1}$, yielding a chart of $\Sigma^*$. 
Note that since $z = \bar z^{-1}$ when $|z|=1$, this is consistent with the fact that we do not want
to change the way in which $A$ and $B$ are identified with the boundaries of $\Sigma$ and $\Sigma^*$.

\begin{remark} 
It is at this point that is is important to choose the metric $g$ in such a way that
$g(dz) = |dz|^2/|z|^2$ in the chart $\psi_j$ since the invariance of this metric under
$z\mapsto \bar z^{-1}$ guarantees that the metric of $\Sigma^*$ still has the same form 
under the charts $\psi_j^*$.
If we had chosen $g$ such that $g(dz) = |dz|^2$ in charts near the boundary, then this 
would not have been the case and attempting to compose $\Sigma^*$ with another morphism 
would have resulted in a manifold equipped with a metric
whose derivative has a jump discontinuity. 
\end{remark}

The second piece of structure is that, given a cobordism $\Sigma\colon A \to A$,
we obtain a cobordism $\tr \Sigma \colon \emptyset \to \emptyset$ by
discarding the degenerate components and gluing $\d\Sigma_\out$ to $\d\Sigma_\inc$
along $A$ in the same way as above, so that $\d\tr\Sigma = \emptyset$. 

\subsection{The Dirichlet form}

Given a cobordism $\Sigma$, we define a bilinear form $\CE_\Sigma$ on smooth functions
$F\colon \Sigma \to \R$ by
\begin{equ}
\CE_\Sigma(F,G) = \int_\Sigma g_z(\nabla F(z),\nabla F(z))\, \Vol_g(dz)\;.
\end{equ}
We note that this expression is independent of the choice of Riemannian metric $g$
in our class. Indeed, when multiplying $g$ by $e^{2\phi}$, $\Vol_g$ gets
multiplied by $e^{2\phi}$ as well (since the real dimension of $\Sigma$ is $2$) and 
$\nabla F$ gets multiplied by $e^{-2\phi}$,
so that $g_z(\nabla F(z),\nabla F(z))$ gets multiplied by $e^{-2\phi}$ and $\CE_\Sigma(F,G)$ doesn't 
change. This is specific to dimension $2$; in dimension $d$, such a conformal change 
of metric would lead to a factor $e^{(d-2)\phi(z)}$ in the integrand.

In particular, we note that if $F$ and $G$ are supported on a single chart $\CU$
which we identify with the corresponding open subset of $\C$, one has
\begin{equ}[e:exprtDFG]
\CE_\Sigma(F,G) = \frac{1}{2\pi}\int_{\CU} \scal{\nabla F(z),\nabla G(z)}\, |dz|^2\;,
\end{equ}
where $\scal{\cdot,\cdot}$ denotes the usual scalar product on $\R^2$ and $|dz|^2$ denotes 
two-dimensional Lebesgue measure. Here, the prefactor $(2\pi)^{-1}$ is to harmonise conventions with the
literature and to get rid of some additional $2\pi$'s later on.

Since the bilinear form $\CE_\Sigma$ is closable and positive, there exists for each choice of
Riemannian metric $g$ a unique positive semidefinite selfadjoint operator $\Delta_g$ such that, 
for all $F$ smooth and $G \in \CD(\Delta_g)$, one has
\begin{equ}[e:defDeltag]
\CE_\Sigma(F,G) = \scal{F,\Delta_g G}_g\;,
\end{equ}
where $\scal{\cdot,\cdot}_g$ denotes the scalar product in $L^2(\Sigma, \Vol_g)$. Integration by 
parts on \eqref{e:exprtDFG} shows that, in local coordinates and when acting on smooth functions that are
supported away from $\d\Sigma$, one has 
$\Delta_g = -(2\pi)^{-1}e^{-2\phi} \Delta$, where $\phi$ is a smooth function that depends on $g$ and the choice
of coordinates, while $\Delta$ is the usual Laplacian on $\R^2$. Regarding boundary conditions,
it follows from an application of the Stokes theorem that, if $\CU$ is as above but this time $F$ and $G$
aren't necessarily supported in $\CU$,
\begin{equs}
\int_{\CU} &\bigl(\scal{\nabla F(z),\nabla G(z)} + F(z)\Delta G(z)\bigr)\, |dz|^2\\
& = 
\int_{\d \CU} F(z)\, (\d_1 G(z)\,dz_2 - \d_2 G(z)\,dz_1)\;.
\end{equs}
(Here, we think of $z = (z_1,z_2)$ as an element of $\R^2$ in the right-hand side.)
Writing again $\Delta_g$ for the differential operator that acts as described above on \textit{all}
smooth functions on $\Sigma$, one finds that this  expression generalises to
\begin{equ}[e:boundary]
\CE_\Sigma(F,G) = \scal{F,\Delta_g G}_g + \frac{1}{2\pi}\oint_{\d\Sigma} F(z) \d_\nu G(z)\,|dz| \;,
\end{equ}
where $\d_\nu G$ denotes the derivative of $G$ in the outward direction normal to $\d\Sigma$. 
We conclude that, in order for \eqref{e:defDeltag} to hold, the domain of 
$\Delta_g$ must consist of functions $G$ such that $\d_\nu G = 0$ on $\d\Sigma$; 
in other words it is the Laplacian
equipped with Neumann boundary conditions. We will make use of the following classical 
facts:

\begin{proposition}
The operator $\Delta_g$ has compact resolvent and its eigenvalues $\lambda_k$ satisfy 
$\lim_{k \to \infty} {\lambda_k \over k} \in (0,\infty)$. Furthermore, its kernel consists
solely of constant functions. \qed
\end{proposition}

\begin{remark}\label{rem:intertwine}
Setting $\tilde g = e^{2\psi} g$, the map $\iota_{g,\tilde g} \colon 
\Phi \mapsto e^{\psi} \Phi$ is an isometry between $L^2(\Sigma, \Vol_{\tilde g})$ and
$L^2(\Sigma, \Vol_g)$. It is then straightforward to verify that
\begin{equ}
\iota_{\tilde g,g} \Delta_g \iota_{g,\tilde g} = 
e^{\psi}\Delta_{\tilde g}e^{\psi}\;.
\end{equ}
\end{remark}

Define now the space $H^1(\Sigma)$ as the completion of the space of smooth functions,
quotiented by constants, under the seminorm $\|F\|_1^2 = \CE_\Sigma(F,F)$. 
Given a metric $g$, we have a canonical compact embedding $H^1(\Sigma) \subset L^2(\Sigma, \Vol_g)$ which maps a class $F$ onto the unique representative such that
$\int F(z) \Vol_g(dz) = 0$. Using this identification and writing $H^{-1}$ for the 
dual of $H^1$, we interpret
elements $\eta \in H^{-1}$ as distributions that vanish on constant test functions.
In this way, the map $F \mapsto \eta_F$ with
$\eta_F(G) = \int F(z) G(z)\,\Vol_g(dz)$ yields a compact embedding $H^1 \subset H^{-1}$.

\begin{remark}
This is quite different from the isomorphism $H^1 \simeq H^{-1}$ given
by the Riesz representation theorem! We will never make use of the latter so
hopefully this will not cause any confusion.
\end{remark}

\subsection{The free field}

Let now $\{e_k\}_{k \ge 1}$ be an orthonormal basis of $H^1(\Sigma)$ consisting of eigenvectors of $\Delta_g$
(ordered by increasing value of the corresponding eigenvalues) and let $e_0$ be the 
function identically equal to one.   
As a consequence of \eqref{e:defDeltag}, these functions are also orthogonal in
$L^2$ and in $H^{-1}$ and one has 
\begin{equ}[e:norms]
\|e_k\|_{L^2}^2 = \lambda_k^{-1} \approx k^{-1}\;,\qquad
\|e_k\|_{H^{-1}}^2 = \lambda_k^{-2} \approx k^{-2}\;.
\end{equ}

Let now $\{\xi_k\}_{k \ge 0}$ be a sequence of i.i.d.\ 
standard Gaussian random variables. We then define the free field measure on $\Sigma$
in the following way. Setting
\begin{equ}[e:defh]
h = \sum_{k \ge 1} \xi_k e_k\;,
\end{equ}
we note that since $\sum_k \|e_k\|_{H^{-1}}^2 < \infty$ by \eqref{e:norms}, this series 
converges in probability in $H^{-1}$ to a random element $h$ of 
$H^{-1}$, the law of which we denote by $\P_\Sigma^{(0)}$. 

The measure $\P_\Sigma^{(0)}$ is Gaussian with covariance 
\begin{equ}
\CG(z,u) = \E h(z)h(u) =  \sum_{k \ge 1} e_k(z)e_k(u)\;,
\end{equ}
that is well-defined for all $z \neq u$. The function $\CG$ is the unique
symmetric function on $\Sigma \times \Sigma$ such that $\CG(z,\cdot)$ and 
$\CG(\cdot,u)$ both have normal derivative vanishing on $\d\Sigma$, 
$\int \CG(z,u)\,\Vol_g(du) = 0$ for every $z$, and 
$\nabla_z\nabla_u \CG(z,u) = 2\pi\delta(z-u)$.
The latter identity automatically holds in 
any (holomorphic) chart since both sides transform in the same way
under conformal transformations in dimension $2$. In particular, 
one has 
\begin{equ}[e:Green]
\CG(z,u) = \log\frac{1}{|z-u|} + \CO(1)\;,
\end{equ}
as $|z-u| \to 0$.
Note that if we make the dependence of $\CG$ on $g$ explicit, one has
\begin{equs}
\CG_{\tilde g}(z,u) &= \CG_g(z,u) - \int \CG_g(z,v)\Vol_{\tilde g}(dv)
 - \int \CG_g(v,u)\Vol_{\tilde g}(dv)\\
 &\qquad  + \int\kern-.7em\int \CG_g(v,w)\Vol_{\tilde g}(dv)\Vol_{\tilde g}(dw)\;,
\end{equs}
thus showing that $h$ depends on $g$ only through a random constant.
One can get rid of this dependence by ``flattening out'' the constant mode in the
following way. Setting
$\tau^{(c)} \colon h \mapsto h+c$,
we define somewhat similarly to before the free field measure 
$\hat \P_\Sigma$ by $\hat \P_\Sigma = \int_{\R^d} \bigl(\tau^{(c)}_* \P^{(0)}_{\Sigma}\bigr)(d\Phi)\,dc$.


\begin{remark}
One can define negative fractional Sobolev spaces $H^{-s}$ by using fractional
powers of $\Delta_g$. One can then see similarly that $h$ defines a random element 
of $H^{-s}$ for every $s > 0$, but not of $L^2$.
\end{remark}

The analogue of the Brownian bridge measure $\hat \P_{2t}$ from the previous section
is then given by $\hat \P_{\hat \Sigma}$, where $\hat \Sigma = \tr(\Sigma \circ \Sigma^*)$.
An important role is being played by the restriction of the free field measure 
$\hat \P_\Sigma$ to the boundary $\d\Sigma$ (or indeed to some smooth simple curve lying in the
interior of $\Sigma$). It is a classical result \cite{Trace} that there is a (unique) bounded
trace operator $\Pi_\Sigma\colon H^1(\Sigma) \to L^2(\d\Sigma)$ extending the usual restriction of continuous 
functions. Since $H^1(\Sigma)$ is the Cameron--Martin space of the Gaussian measure $\hat \P_\Sigma$ 
(ignoring here the inessential complication coming from the fact that this measure is only
$\sigma$-finite due to the constant mode), it follows from standard Gaussian measure theory
results \cite[Sec.~4.3]{IntroSPDEs} that $\Pi_\Sigma$ extends uniquely to a linear subspace of full $\hat \P_\Sigma$-measure,
thus yielding a Gaussian measure $\hat \P_{\d\Sigma} = (\Pi_\Sigma)_* \hat \P_\Sigma$ on any Hilbert space $\CH$ containing $L^2$ 
and such that the embedding $L^2 \hookrightarrow \CH$ is Hilbert--Schmidt, for example $\CH = H^{-1}(\d\Sigma)$
will do. 

In order to describe the measure $\hat \P_{\d\Sigma}$, we first note that any smooth function $\Phi\colon \Sigma \to \R$
can be written uniquely as $\Phi = \Phi_0 + \Phi_\d$ where $\Phi_0$ vanishes on $\d\Sigma$ while 
$\Phi_\d$ is harmonic in the interior of $\Sigma$. In particular, it follows from \eqref{e:boundary}
that one has $\CE_\Sigma(\Phi_0, \Phi_\d) = 0$, so that this yields a decomposition of $H^1(\Sigma)$ into 
orthogonal subspaces:
\begin{equ}[e:decomp]
H^1(\Sigma) = H^1_0(\Sigma) \oplus H^1_\d(\Sigma)\;.
\end{equ}
One furthermore has
\begin{equ}
\CE_\Sigma(\Phi_\d,\Phi_\d) = \frac{1}{2\pi}\oint_{\d\Sigma} \Phi(z)\,\d_\nu \Phi_\d(z)\,|dz|\;,
\end{equ}
which naturally leads to the introduction of the Dirichlet to Neumann operator
$D_\Sigma \colon \Phi \mapsto \d_\nu \Phi_\d$. While $\Phi$ is defined on all of $\Sigma$, 
$\Phi_\d$ only depends on the restriction of $\Phi$ to $\d\Sigma$, so $D_\Sigma$ is naturally
interpreted as an unbounded operator on $L^2(\d\Sigma)$.

To understand this operator, consider the simplest possible case, namely that of 
a disk $\Sigma = D = \{z\,:\, |z| \le 1\}$, so that $\d\Sigma \simeq S^1$. 
The harmonic extension of $e_n \colon \theta \mapsto \cos(n\theta)$ to $D$ is then 
given by $z \mapsto \Re z^n$, and similarly for $e_n^*(\theta) = \sin(n\theta)$ which extends to $\Im z^n$.
Along any ray $r \mapsto re^{i\theta}$, these functions are proportional to $r^n$, whence we conclude that
\begin{equ}
D_\Sigma e_n = n e_n\;,\qquad D_\Sigma e_n^* = n e_n^*\;.
\end{equ}
In particular $D_\Sigma = \sqrt{-\Delta}$, so that the space $H^1_\d(D)$ introduced in \eqref{e:decomp}
does in fact coincide with the fractional Sobolev space $H^{1/2}(S^1)$. 

It follows that in this particular case the restriction of the free field to $\d\Sigma$ yields
a random distribution which (ignoring the constant mode) 
can be written as the random Fourier series
\begin{equ}[e:GFFcircle]
h = \sum_{n \ge 1} \sqrt{\frac2{n}}\bigl(\xi_n e_n + \xi_n^* e_n^*\bigr)\;,
\end{equ}
where the $\xi_n$ and $\xi_n^*$ are two independent sequences of i.i.d.\ normal random variables.
Here, the $2$ comes from the fact that $e_n$ has $L^2$ norm equal to $1/\sqrt{2}$.

\begin{remark}
While the case of general $\Sigma$ does of course yield something different, it turns out that 
the measure $\hat \P_{\d\Sigma}$ is always equivalent to the law of independent copies of $h$ as in \eqref{e:GFFcircle},
one for each connected component of $\d\Sigma$. 
\end{remark}

\subsection{Segal's axioms}
\label{sec:Segal}

We now discuss how to generalise the construction of \eqref{e:exprPV}
to this two-dimensional setting and in particular how to obtain the analogue
of the Markov property as in Corollary~\ref{cor:SG}. 
For this we assume that, for every cobordism $\Sigma$, we are given a measurable
function $V_\Sigma \colon H^{-1}(\Sigma) \to \R$ defined modulo $\hat \P_\Sigma$-null sets,
with the following properties:
\begin{itemize}
\item \textbf{Locality.} If $\Sigma = \Sigma_1 \circ\Sigma_2$, writing $\Pi_{k} \colon H^{-1}(\Sigma) \to H^{-1}(\Sigma_k)$ for the restriction operator, one has the $\hat \P_\Sigma$-almost sure identity $V_\Sigma(\Phi) = V_{\Sigma_1}(\Pi_1\Phi) + V_{\Sigma_2}(\Pi_2\Phi)$. The same holds if $\Sigma = \Sigma_1 \otimes \Sigma_2$.
\item \textbf{Coercivity.} For every cobordism $\Sigma$ (without degenerate components), the measure $\exp(-V_\Sigma(\Phi))\,\hat \P_\Sigma(d\Phi)$
is finite.
\end{itemize}

\begin{remark}
Since the cobordism $\Sigma$ is already implicit in the field $\Phi$, we will
usually simply write $V(\Phi)$ rather than $V_\Sigma(\Phi)$.
\end{remark}

We then define as in \eqref{e:exprPV}
\begin{equ}[e:defPVtwo]
\P_\Sigma = (\det_\zeta \Delta_g)^{-1/2}\hat \P_{\hat \Sigma}\;,\qquad 
\P^V_\Sigma(d\Phi) = \exp \bigl(-V(\Phi)\bigr)\,\P_\Sigma(d\Phi)\;.
\end{equ}
Here, the operator $\Delta_g$ is as in \eqref{e:defDeltag}, but for the 
manifold without boundary $\hat \Sigma$.

The analogue of the construction of Section~\ref{sec:cobord1} is then again a monoidal functor
$\CF^V$, but this time from $\CC^{(2)}$ to $\Hil$, defined as follows. 
Given any $A \in \Ob \CC^{(2)}$, we write $\hat \P_{A}$ for the measure on $H^{-1}(A)$ 
given by the law of independent copies of $h$ as in \eqref{e:GFFcircle}, one for each connected
component of $A$ (which are isometric to $S^1$). Note that the isometry $A \to (S^1)^{|A|}$
with $|A|$ the number of connected components of $A$ is not canonical since we can rotate
each of these circles. However, the law of $h$ is invariant under such rotations,
so that the measure $\hat \P_{A}$ is well defined. 
We then set $\CF^V(A) = \CH_A$, the space of half-densities on $H^{-1}(A)$ associated to the
measure class of $\hat \P_{A}$. One has canonical identifications $H^{-1}(A \sqcup B) = H^{-1}(A)\times H^{-1}(B)$
and $\hat \P_{A\sqcup B} \simeq \hat \P_{A} \otimes \hat \P_{B}$, so that $\CH_{A\sqcup B} \simeq \CH_A \otimes \CH_B$,
showing that $\CF^V$ respects the monoidal structure.

Let now $\Sigma \colon A\to B$ be a cobordism.
By Remark~\ref{rem:tensor} and the coercivity property of $V$ which guarantees that the measure $\P^V_\Sigma$ is finite, 
\begin{equ}
\D^V_\Sigma \eqdef \sqrt{(\Pi_\Sigma)_* \P^V_\Sigma}
\end{equ}
then yields an element
of $\CH_{\d\Sigma} \simeq \CH_A \otimes \CH_B$, similarly to what we just discussed, thus defining
a Hilbert--Schmidt operator $\CF^V(\Sigma)\colon \CH_A \to \CH_B$. 

\begin{remark}
In the case of degenerate components,
note that the rotation operators $\tau^{(\theta)} \colon H^{-1}(S^1)\to H^{-1}(S^1)$
induce rotation operators $\mu \mapsto \tau^{(\theta)}_* \mu$ acting on measures on $H^{-1}(S^1)$. Since 
the measure $\hat \P_{S^1}$ is invariant under rotations, this in turn yields a unitary action of the 
rotation group on $\CH_{S^1}$. This allows one to naturally associate unitary operators to degenerate
components of $\Sigma$.
\end{remark}

The following theorem is essentially a simplified version of some of the results 
of \cite{SegalSimple,Bootstrap}. 

\begin{theorem}
$\CF^V$ as defined above is a monoidal functor  $\CC^{(2)} \to \Hil$.
\end{theorem}

\begin{proof}[Sketch of proof]
We proceed at the formal level, ``pretending'' that the spaces $H^{-1}(\Sigma)$ are finite-dimensional
and that the $\zeta$-regularised determinants behave like ``real'' determinants. In this way, we hope to 
highlight the essence of the argument without getting bogged down in technicalities.

We fix $\Sigma=\Sigma_2 \circ \Sigma_1$, the composition of two cobordisms 
$\Sigma_1\colon A_1 \to A_2$ and $\Sigma_2\colon A_2 \to A_3$. We assume for simplicity that neither has 
a degenerate component and that every connected component of $\Sigma$ has a non-trivial boundary. 
Our aim is then to show that $\CF^V(\Sigma) = \CF^V(\Sigma_2)\circ \CF^V(\Sigma_1)$.
We note that, at a formal level, one has
\begin{equ}
\P^V_{\Sigma_1}(d\Phi) = \exp \Bigl(- \f12 \CE_{\hat \Sigma_1}(\Phi,\Phi)- V_{\hat \Sigma_1}(\Phi)\Bigr)\,d\Phi\;,
\end{equ}
since the first part of \eqref{e:defPVtwo} cancels out the determinant that would otherwise multiply the 
Gaussian density. Here, even though $\P^V_{\Sigma_1}$ is a measure on $H^{-1}(\hat \Sigma_1)$, we think
of $\Phi$ as taking values in $H^1(\hat \Sigma_1)$, at least as far as the first term in the exponent is concerned.

At this point, we note that since $\hat \Sigma_1$ consists of two copies of $\Sigma_1$ glued together
symmetrically along their joint boundary, we have the decomposition
\begin{equ}[e:decompH1]
H^1(\hat \Sigma_1) = H^1_0(\Sigma_1) \oplus H^1_0(\Sigma_1) \oplus H^{1/2}(\d\Sigma_1)\;.
\end{equ}
This corresponds to the decomposition of any smooth function $\Phi \colon\hat \Sigma_1 \to \R$
as
\begin{equ}
\Phi = \Phi_0^{(1)} + \Phi_0^{(2)} + E (\Phi \restr \d\Sigma_1) \;,
\end{equ}
(here $\restr$ denotes restriction) where $E_{\hat \Sigma_1} \colon H^{1/2}(\d\Sigma_1) \to H^1(\hat \Sigma_1)$ is the harmonic extension to 
$\hat \Sigma_1 \setminus \d\Sigma_1$,
$\Phi_0^{(1)}$ is supported on $\Sigma_1$ (and vanishes on $\d\Sigma_1$), 
and $\Phi_0^{(2)}$ is supported on $\Sigma_1^*$ (which is canonically identified with $\Sigma_1$ as a set). 
Note that furthermore
\begin{equ}
\CE_{\hat \Sigma_1}(\Phi,\Phi) =  \CE_{\Sigma_1}^{(0)}(\Phi_0^{(1)},\Phi_0^{(1)}) + \CE_{\Sigma_1}^{(0)}(\Phi_0^{(2)},\Phi_0^{(2)})
+ \frac{1}{\pi} \oint_{\d\Sigma_1} \Phi(z) D_{\Sigma_1} \Phi(z)\,|dz|\;.
\end{equ}
Here, the factor $\frac{1}{\pi}$ as opposed to $\frac{1}{2\pi}$ in the last term 
comes from the fact that we have one contribution from each of the
two copies of $\Sigma_1$ glued along $\d\Sigma_1$. We also write $\CE_{\Sigma_1}^{(0)}$ to emphasise that this
bilinear form enforces Dirichlet boundary conditions. 
Since furthermore the locality property of $V$ implies that
\begin{equ}
V_{\hat \Sigma_1}(\Phi) = V_{\Sigma_1}(\Phi_0^{(1)} + E_{\Sigma_1}\Phi) + V_{\Sigma_1}(\Phi_0^{(2)} + E_{\Sigma_1}\Phi)\;,
\end{equ}
this strongly hints that we can write
\begin{equs}
(\Pi_{\Sigma_1})_*\P^V_{\Sigma_1}(d\Phi) 
 &= \biggl(\int_{(\Sigma_1)} \exp \Bigl(-\frac12 \CE_{\Sigma_1}^{(0)}(\Psi,\Psi)  -V_{\Sigma_1}(\Psi + E_{\Sigma_1} \Phi)\Bigr)\,d\Psi \biggr)^2\\
 &\qquad \times \exp \Bigl(- \frac{1}{2\pi}\scal{\Phi, D_{\Sigma_1} \Phi}_{\d\Sigma_1}\Bigr)\,d\Phi\;,
\end{equs}
where now $\Phi$ denotes an element of $H^{1/2}(\d\Sigma_1)$ and $\scal{\cdot,\cdot}_{\d\Sigma_1}$ denotes
the usual scalar product of $L^2(\d\Sigma_1)$. We write $(\Sigma_1)$ as the domain of integration to 
remind ourselves that we are integrating over some space of functions / distributions on $\Sigma_1$, but we are 
intentionally vague as to what this space is exactly.
This yields
\begin{equs}
\D^V_{\Sigma_1}(d\Phi)
 &= \int_{(\Sigma_1)} \exp \Bigl(-\f12 \CE_{\Sigma_1}^{(0)}(\Psi,\Psi)  - V_{\Sigma_1}(\Psi + E_{\Sigma_1} \Phi)\Bigr)\,d\Psi\\
 &\qquad\times \exp \Bigl(- \frac1{4\pi}\scal{\Phi, D_{\Sigma_1} \Phi}_{\d\Sigma_1}\Bigr)\,\sqrt{d\Phi}\label{e:weight}\\
 &\eqdef \CD_{\Sigma_1}^V(\Phi) \exp \Bigl(- \frac1{4\pi}\scal{\Phi, D_{\Sigma_1} \Phi}_{\d\Sigma_1}\Bigr)\,\sqrt{d\Phi}\;.
\end{equs}
Since $\d\Sigma_1 \simeq A_1 \sqcup A_2$, we can write $\Phi = (\Phi_1,\Phi_2)$ with 
$\Phi_k \in H^{1/2}(A_k)$. With this notation, our goal is to show that one has
\begin{equ}[e:wantedIdentity]
\int \D^V_{\Sigma_1}(d\Phi_1,d\Phi_2)\,\D^V_{\Sigma_2}(d\Phi_2,d\Phi_3)
= \D^V_{\Sigma}(d\Phi_1,d\Phi_3)\;,
\end{equ}
where the integration runs over the variable $\Phi_2$.

Let us now for the moment concentrate on the factor appearing on the second line of \eqref{e:weight};
let's call it $\D_{\Sigma_1}(d\Phi)$.
 Write furthermore $D_{j,k}^{(1)} \Phi_k$ as a shorthand for the function on
$A_j$ given by $\d_\nu E_{j,k}^{(1)}\Phi_k$, where $E_{j,k}^{(\ell)} \Phi_k$ is the harmonic function on
$\Sigma_\ell$ which agrees with $\Phi_k$ on $A_k$ and vanishes on $\d\Sigma_1 \setminus A_k$. 
With these notations, one then has
\begin{equ}
\scal{\Phi, D_{\Sigma_1} \Phi}_{\d\Sigma_1}
= \scal{\Phi_1,D_{1,1}^{(1)}\Phi_1}_{A_1} + \scal{\Phi_2,D_{2,2}^{(1)}\Phi_2}_{A_2}
+ 2\scal{D_{2,1}^{(1)}\Phi_1,\Phi_2}_{A_1}\;.
\end{equ}
(Note that the adjoint of $D_{1,2}^{(1)}$ is $D_{2,1}^{(1)}$ by symmetry.)
Writing 
\begin{equ}
D_{2,2} = D_{2,2}^{(1)} + D_{2,2}^{(2)}\;,
\end{equ}
as well as 
\begin{equ}
M\Phi = D_{2,2}^{-1}\bigl(D_{2,1}^{(1)}\Phi_1+D_{2,3}^{(2)}\Phi_3\bigr)\;,
\end{equ}
it follows that 
\begin{equs}
\D_{\Sigma_1}&(d\Phi_1,d\Phi_2)\D_{\Sigma_2}(d\Phi_2,d\Phi_3) \label{e:productDs}\\
&= \exp \Bigl(-\frac1{4\pi} \scal{\Phi_2+M\Phi,D_{2,2}(\Phi_2 +M\Phi)}_{A_2} \Bigr)d\Phi_2 \sqrt{d\Phi_1d\Phi_3}\\
&\quad \times \exp \Bigl(-\frac1{4\pi}\scal{\Phi_1,D_{1,1}^{(1)}\Phi_1}_{A_1} -\frac1{4\pi} \scal{\Phi_3,D_{3,3}^{(2)}\Phi_3}_{A_3}
+ \frac1{4\pi} \scal{M\Phi,D_{2,2} M\Phi}_{A_2}\Bigr)\;.
\end{equs}

\begin{remark}\label{rem:jump}
It is not obvious a priori that $D_{2,2}$ is surjective so that $M$ is well defined. 
Observe though that $D_{2,2}\Phi_2$ 
is nothing but the jump in the normal derivative across $A_2$ of the harmonic extension $E\Phi_2$ of $\Phi_2$ to 
$\Sigma \setminus A_2$ with Dirichlet boundary conditions on $\d\Sigma$. In order to have $D_{2,2}\Phi_2 = 0$,
that harmonic extension must be smooth across $A_2$ and therefore be harmonic on all of $\Sigma$. However,
since every connected component of $\Sigma$ has a non-trivial boundary, there is no non-zero harmonic function 
on $\Sigma$ vanishing on $\d\Sigma$.

In fact, even if we have a connected component of $\Sigma$ without boundary intersecting $A_2$, then
$D_{2,1}^{(1)}\Phi_1$ and $D_{2,3}^{(2)}\Phi_3$ vanish there so we can still define $M$ canonically.
\end{remark}

Let us now examine in more detail the last line in \eqref{e:productDs}. We can rewrite the bilinear form
in the exponential as
\begin{equ}[e:formexp]
\scal{\Phi_1,D_{1,1}^{(1)}\Phi_1 - D_{1,2}^{(1)}\Xi}_{A_1} + \scal{\Phi_3,D_{3,3}^{(2)}\Phi_3 - D_{3,2}^{(2)}\Xi}_{A_3}
\end{equ}
where $\Xi = M\Phi$ is the unique solution to
\begin{equ}[e:defXi]
D_{2,2}\Xi = \bigl(D_{2,1}^{(1)}\Phi_1+D_{2,3}^{(2)}\Phi_3\bigr)\;.
\end{equ}
In other words, the harmonic extension $E\Xi$ of $\Xi$ to $\Sigma \setminus A_2$ vanishing on $A_1 \cup A_3$ 
is such that the jump
in its normal derivative at $A_2$ equals that of the function that agrees with $E_{2,1}^{(1)}\Phi_1$ on $\Sigma_1$ 
and with $E_{2,3}^{(2)}\Phi_3$ on $\Sigma_2$. This implies that \eqref{e:formexp} is nothing but
\begin{equ}
\scal{(\Phi_1,\Phi_3),D_\Sigma(\Phi_1,\Phi_3)}_{A_1 \cup A_3}\;,
\end{equ}
so that \eqref{e:productDs} can be rewritten as
\begin{equs}
\D_{\Sigma_1}&(d\Phi_1,d\Phi_2)\D_{\Sigma_2}(d\Phi_2,d\Phi_3) \\
&= \exp \Bigl(-\frac1{4\pi} \scal{\Phi_2+M\Phi,D_{2,2}(\Phi_2 +M\Phi)}_{A_2} \Bigr)d\Phi_2\, \D_{\Sigma}(d\Phi_1,d\Phi_3)\;.
\end{equs}

In order to show the desired functorial property, it therefore remains to show that 
\begin{equs}
\CD_{\Sigma}^V(\Phi_1,\Phi_3) &=
\int_{(A_2)} \CD_{\Sigma_1}^V(\Phi_1,\Phi_2) \CD_{\Sigma_2}^V(\Phi_2,\Phi_3) \exp \Bigl(-\frac1{4\pi} \scal{\Phi_2+M\Phi,D_{2,2}(\Phi_2 +M\Phi)}_{A_2} \Bigr)d\Phi_2 \\
&=\int_{(A_2)} \CD_{\Sigma_1}^V(\Phi_1,\Phi_2-M\Phi) \CD_{\Sigma_2}^V(\Phi_2-M\Phi,\Phi_3) \exp \Bigl(-\frac1{4\pi} \scal{\Phi_2,D_{2,2}\Phi_2}_{A_2} \Bigr)d\Phi_2 \;.
\end{equs}

At this stage, as a consequence of Remark~\ref{rem:jump} and a reasoning very similar to the one for \eqref{e:decompH1}, 
we realise that one has the decomposition
\begin{equ}
H_0^1(\Sigma) = H_0^1(\Sigma_1) \oplus H_0^1(\Sigma_2) \oplus H^{1/2}(A_2)\;,
\end{equ}
with $H^{1/2}(A_2) \subset H_0^1(\Sigma)$ via the harmonic extension $E$ as in the remark, and that 
with this identification $\CE_{\Sigma}^{(0)}$
coincides with $\f1{2\pi}\scal{\cdot,D_{2,2}\cdot}_{A_2}$ on that last component.

As a consequence, we can write
\begin{equs}
\CD_{\Sigma}^V(\Phi_1,\Phi_3) &= 
\int \exp \Bigl(-\frac12 \CE_{\Sigma_1}^{(0)}(\Psi_1,\Psi_1)-\frac12 \CE_{\Sigma_2}^{(0)}(\Psi_2,\Psi_2)
- \frac1{4\pi}\scal{\Phi_2,D_{2,2}\Phi_2}_{A_2}\Bigr)\\
&\times \exp \Bigl(- V_{\Sigma_1}\bigl(\Psi_1 + E \Phi_2 + E_\Sigma(\Phi_1,\Phi_3)\bigr)\Bigr)\\
&\times \exp \Bigl(- V_{\Sigma_2}\bigl(\Psi_2 + E \Phi_2 + E_\Sigma(\Phi_1,\Phi_3)\bigr)\Bigr)d\Psi_1 d\Psi_2d\Phi_2\;,
\end{equs}
where the arguments of $V_{\Sigma_1}$ and $V_{\Sigma_2}$ are restricted to the appropriate manifold.
It therefore remains to show that
\begin{equ}[e:finalIdentity]
\bigl(E \Phi_2 + E_\Sigma(\Phi_1,\Phi_3)\bigr) \restr \Sigma_1 
= E_{\Sigma_1} \bigl(\Phi_1,\Phi_2 - M\Phi\bigr)\;,
\end{equ}
and the analogous property for $\Sigma_2$. Both functions are harmonic and agree with $\Phi_1$ on $A_1$,
so it remains to show that their values on $A_2$ agree. Since $E \Phi_2$ and $E_{\Sigma_1}(0,\Phi_2)$ agree on $\Sigma_1$, we can assume that $\Phi_2 = 0$ without loss of generality.

Write now $E_{\Sigma_1,\Sigma_2}(\Phi_1,\Phi_2,\Phi_3)$ for the function that is harmonic
on $\Sigma \setminus (A_1 \cup A_2 \cup A_3)$ and agrees with $\Phi_k$ on $A_k$.
It then follows from the discussion after \eqref{e:defXi} that the normal derivatives of
$E_{\Sigma_1,\Sigma_2}(\Phi_1,0,\Phi_3)$ and $E_{\Sigma_1,\Sigma_2}(0,M\Phi,0)$
have the same jump discontinuity across $A_2$. By linearity, it then follows that
$E_{\Sigma_1,\Sigma_2}(\Phi_1,-M\Phi,\Phi_3)$ is smooth across $A_2$ and therefore has to agree with $E_\Sigma(\Phi_1,\Phi_3)$ on all of $\Sigma$. Since on the other hand
$E_{\Sigma_1,\Sigma_2}(\Phi_1,-M\Phi,\Phi_3) = E_{\Sigma_1}(\Phi_1,-M\Phi)$ on $\Sigma_1$,
we have shown that the identity \eqref{e:finalIdentity} holds, thus completing the proof.
\end{proof}

\section{Liouville theory}

If we want $V$ to satisfy the locality property, it is natural to consider 
expressions of the form
\begin{equ}[e:localV]
V_\Sigma(\Phi) = \int_\Sigma V(\Phi(z))\,\Vol_g(dz)\;.
\end{equ}
Unfortunately, this appears nonsensical since we furthermore require $V_\Sigma$ to
be defined for $\hat \P_\Sigma$-almost every $\Phi$ and this measure is only
supported on spaces of distributions where point evaluations aren't defined. 
It is however possible to remedy this problem in the following way. 

As already discussed in the previous section,  $\hat \P_\Sigma$-almost every $\Phi$
can be restricted to a smooth curve. Given $z \in \Sigma$, we then write 
$\Gamma_\eps(z)$ for the ``circle of radius $\eps$'', i.e.\ the set of points 
at distance  $\eps$ (with respect to the metric $g$) of $z$. We can then define
\begin{equ}[e:defPhieps]
\Phi_\eps(z) = |\Gamma_\eps(z)|^{-1}\int_{\Gamma_\eps(z)} \Phi(u)\,|du|_g\;,
\end{equ}
where $|\Gamma_\eps(z)| = \int_{\Gamma_\eps(z)}|du|_g$ denotes the perimeter of the circle.
As a consequence of \eqref{e:Green}, one then has
\begin{equ}
\E |\Phi_\eps(z)|^2 = |{\log\eps}| + \CO(1)\;,
\end{equ}
where furthermore, since the law of $\Phi_\eps$ only depends on $g$ via its constant mode,
the $\CO(1)$ term is independent of the choice of $g$, except possibly for a constant
(i.e.\ independent of $z$) term. See for example \cite{MR3465434} for more details
in the particular case when $\Sigma$ is a sphere.

So while there is no hope for $\Phi_\eps^2(z)$ to converge to a limit, even in the sense
of distributions, as $\eps \to 0$, one can show 
that the ``Wick square''
\begin{equ}
\Wick{\Phi_\eps^2}(z) \eqdef \Phi_\eps^2(z) - |{\log\eps}|\;,
\end{equ}
converges $\hat \P_\Sigma$-almost surely to a limiting distribution 
$\Wick{\Phi^2}$. The same is true also for higher Wick powers, defined by
\begin{equ}
\Wick{\Phi_\eps^n}(z) \eqdef H_n(\Phi_\eps(z), |{\log\eps}|)\;.
\end{equ}
It is then possible to set for example $V(\Phi) = \int_\Sigma \Wick{\Phi^4}(z)\,\Vol_g(dz)$ in order to obtain a local and coercive potential for which the construction
in Section~\ref{sec:Segal} can be carried out, see for example \cite{Nelson1,SegalSimple}. 

\subsection{Conformal changes of metric}

It is apparent from \eqref{e:defPhieps} that the distributions $\Wick{\Phi^n}$ 
obtained in this way are independent of the chosen coordinate system. They do however
transform non-trivially under conformal changes of metric! Indeed, if 
$\tilde g = e^{2\phi} g$
as before and writing $\Phi_{\eps,g}$ to make the dependence on $g$ explicit, one
has $\Gamma_{\eps,\tilde g}(z) \approx \Gamma_{e^{-\phi(z)}\eps,g}(z)$ for small values
of $\eps$ so that one would expect that
$\Phi_{\eps,\tilde g}(z) \approx \Phi_{e^{-\phi(z)}\eps,g}(z)$, 
yielding the almost sure identity $\Wick{\Phi^2_{\tilde g}} = \Wick{\Phi^2_{g}} + \phi$.
In the case of $\Wick{\Phi^4}$, we get additional terms proportional to $\Wick{\Phi^2}$,
and so we have genuinely different theories for every choice of metric $g$.

Liouville theory on the other hand transforms in a much ``nicer way'', and this is
what we would like to explain now. Choose $\mu, \gamma > 0$ (we will see later on that we
in fact have to take $\gamma < 2$) and set
\begin{equ}
Q = \frac\gamma2 + \frac2\gamma\;.
\end{equ}
We then set
\begin{equ}[e:VLiouville]
V_\eps^g(\Phi) = \int_\Sigma \Bigl(\frac{Q}{4\pi}R_g \Phi_\eps(z) + \mu\eps^{\frac{\gamma^2}{2}}e^{\gamma \Phi_\eps(z)} \Bigr)\,\Vol_g(dz)\;,
\end{equ}
where $R_g$ denotes the scalar curvature of the metric $g$.
Note first that it is at least somewhat plausible that this expression has a limit as
$\eps \to 0$ since, under $\hat\P_\Sigma$, we have by Gaussianity 
\begin{equ}
\E e^{\gamma \Phi_\eps(z)} = \exp \Bigl(\frac{\gamma^2}{2} \E \Phi_\eps^2(z)\Bigr)
= \exp \Bigl(\frac{\gamma^2}{2} |{\log\eps}| + \CO(1)\Bigr)
\approx \eps^{-\frac{\gamma^2}{2}}\;.
\end{equ}
The reason why $\gamma = 2$ is a threshold is
that one expects to have 
\begin{equ}
\sup_{z\in \Sigma} \Phi_\eps(z) = 2\log \eps + o(\log\eps)\;.
\end{equ}
This is because a relatively good ``cartoon'' for $\Phi_\eps$ is that of 
a collection of i.i.d.\ Gaussians with variance $|{\log\eps}|$ distributed on a 
grid of mesh size $\CO(\eps)$. Since the supremum of $N$ unit variance Gaussians
is given by $\sqrt{2\log N} + o(\sqrt{\log N})$, this suggests that one
has indeed 
\begin{equ}
\sup_{z\in \Sigma} \Phi_\eps(z)
\approx \sqrt{2\log (K\eps^{-2})} \sqrt{|{\log\eps}|} \approx 2|{\log\eps}|\;.
\end{equ}
We conclude that the tallest peaks of $e^{\gamma \Phi_\eps(z)}$ are of 
height about $\eps^{-2\gamma}$. Since one expects these peaks to have a diameter of
order $\eps$, any such peak would contribute about $\eps^{\frac{\gamma^2}{2}+2-2\gamma}$
to the integral in \eqref{e:VLiouville}. At $\gamma = 2$, this is $\CO(1)$ so that 
one would expect a drastic change in behaviour, which is indeed the case,
see for example \cite{Critical_der,Glassy}.

One feature of \eqref{e:VLiouville} is that its
limit behaves in an interesting way under conformal changes of metric.
Indeed, setting $\tilde g = e^\phi g$ as usual, we recall that the scalar curvature 
transforms under conformal changes of metric like
\begin{equ}[e:changeRg]
R_{\tilde g} = e^{-2\phi} \bigl(R_g - 2\Delta^{(g)} \phi\bigr)\;,
\end{equ}
where $\Delta^{(g)}$ denotes the usual Laplace--Beltrami operator with respect to
the metric $g$ (which equals $-2\pi \Delta_g$ defined in the previous section).
See for example \cite[Eq.~2.6]{LP87}, noting that their sign convention for the Laplacian is
opposite to the one used in \eqref{e:changeRg}.

Writing $\tilde\eps = e^{-\phi(z)}\eps$, one then has
\begin{equs}
V_\eps^{\tilde g}(\Phi) &= \int_\Sigma \Bigl(\frac{Q}{4\pi}\bigl(R_g +4\pi\Delta_g \phi\bigr) \Phi_{\eps,\tilde g}(z) + \mu\eps^{\frac{\gamma^2}{2}}e^{\gamma \Phi_{\tilde\eps,g}(z)+2\phi(z)} \Bigr)\,\Vol_g(dz)\;.
\end{equs}
Note now that on one hand one has
\begin{equ}
\int_\Sigma \Delta_g \phi(z) \Phi(z)\,\Vol_g(dz) = \CE_\Sigma(\phi,\Phi)\;,
\end{equ}
and on the other hand
\begin{equ}
\eps^{\frac{\gamma^2}{2}}e^{\gamma \Phi_{\tilde \eps,g}(z)+2\phi(z)}
=
{\tilde \eps}^{\frac{\gamma^2}{2}}e^{\gamma \Phi_{\tilde \eps,g}(z)+ \gamma Q \phi(z)}\;.
\end{equ}
Assuming that the limit $V^{g} = \lim_{\eps \to 0} V_\eps^g$ exists,
this shows that one has
\begin{equs}
\frac12 \CE_\Sigma(\Phi,\Phi) + V^{\tilde g}(\Phi)
&= \frac12 \CE_\Sigma(\Phi,\Phi) + V^{g}(\Phi + Q\phi)
- \frac{Q^2}{4\pi} \int R_g(z)\phi(z)\Vol_g(dz) \\
&\quad +  \CE_\Sigma(Q\phi,\Phi) \\
&= \frac12 \CE_\Sigma(\Phi + Q\phi,\Phi + Q\phi) + V^{ g}(\Phi + Q\phi) \\
&\quad - \frac{6Q^2}{24\pi} \Bigl(\int R_g(z)\phi(z)\Vol_g(dz) 
+2\pi \CE_\Sigma(\phi,\phi)\Bigr)\;.
\end{equs}
In other words, modulo the $\Phi$-independent factor appearing on the last line,
a conformal change of metric is equivalent to a simple shift of the field by $Q\phi$,
as well as a change in normalisation of the measure by some expression that only depends on $\phi$.
The reason for writing this last term in this way is conform to the literature
where this is called the ``conformal anomaly'' and the factor $6Q^2$ appearing there is called the 
``central charge'' of the theory. In fact, the central charge of Liouville theory happens to
be $1+6Q^2$ (and not just $6Q^2$) as a consequence of the fact that the exponential of
the same term, but without the factor $6Q^2$, appears as the logarithm of the ratio of 
$\sqrt{\det_\zeta \Delta_g}$ and $\sqrt{\det_\zeta \Delta_{\tilde g}}$,
see for example \cite{OPS}.

\subsection{Coercivity of Liouville theory}

We now argue under what conditions one can make Liouville theory
coercive, so that it falls (almost) into the framework developed in the 
previous section. Assuming that the scalar curvature $R$ of $g$ is constant,
we see that at least formally $V^g$ is of the form \eqref{e:localV} with
\begin{equ}
V(u) = \frac{QR}{4\pi} u + \mu e^{\gamma u}\;.
\end{equ}
Since $\gamma > 0$ this function grows at $+\infty$, but it only grows
at $-\infty$ when $R<0$. This suggests that Liouville theory is coercive only
when $g$ is conformally equivalent to a metric with negative scalar curvature.
By the uniformisation theorem, this is the case precise when 
$\Sigma$ has genus at least $2$.

If the genus of $\Sigma$ is less than $2$, we can insert conical singularities
which is achieved by adding point masses to the scalar curvature. 
Let us assume that $\Sigma$ is a sphere. By the previous discussion, combined
with the uniformisation theorem, we can reduce ourselves to the case of the round sphere
of area $1$, so that $R_g = 8\pi$.
Given
values $\alpha_i \in \R$ and $z_i\in\Sigma$, we would therefore be tempted to define
\begin{equ}
V_{(\alpha,z)}^{\tilde g}(\Phi) = V^{\tilde g}(\Phi)
- \sum_{i=1}^k \alpha_i \Phi(z_i)\;.
\end{equ}
(These are called \textit{insertions} and the parameter $\alpha_i$ is called the
\textit{weight} of the $i$th insertion.)
As already mentioned, we cannot consider point evaluations of the free field.
However, we have at least formally the identity
\begin{equ}
\Phi(z_i) = \CE_\Sigma\bigl(\Phi, \CG_{z_i}\bigr) + \int_\Sigma \Phi(z)\,\Vol_g(dz)\;,
\end{equ}
where $\CG_u$ is such that $\Delta_g \CG_u = \delta_u - 1$. 
(Here we have to add the constant term to guarantee that the right-hand side
of the Poisson equation has vanishing mean, which guarantees its solvability.)
As before, the resulting term $\sum_i \alpha_i \CE_\Sigma\bigl(\Phi, \CG_{z_i}\bigr)$ 
can be eliminated by a Girsanov shift, leading to a potential in the shifted
variables given by 
\begin{equ}[e:VLiouvilleInsertions]
V_{(\alpha,z)}^g(\Phi) = \lim_{\eps \to 0}\int_\Sigma \Bigl(\Bigl(2Q - \sum_{i=1}^k \alpha_i\Bigr) \Phi(z) + \mu\eps^{\frac{\gamma^2}{2}}e^{\gamma (\Phi_\eps + \sum_i \alpha_i \CG_{z_i})(z)} \Bigr)\,\Vol_g(dz)\;.
\end{equ}
This suggests very strongly that, in order to obtain a finite measure, 
one should impose the Seiberg bounds $\sum_{i=1}^k \alpha_i > 2Q$.
This is indeed the case \cite{MR3465434} and it turns out that the resulting
theory has functorial properties similar to those verified in the previous
section, except that the Riemann surfaces $\Sigma$ are furthermore equipped
with marked points $z_i$ to which weights $\alpha_i$ are attached \cite{Bootstrap}.

\subsection*{Acknowledgements}

{\small
I am grateful to Nicolas Bourbaki and Juhan Aru for reading through a draft version of 
these notes and pointing out a number of typos and imprecisions. 
}

\bibliographystyle{Martin}
\bibliography{Bourbaki}

\end{document}